\documentclass[aps,reprint,superscriptaddress,pra,10pt]{revtex4-1}
\usepackage[utf8]{inputenc}
\usepackage[T1]{fontenc}
\usepackage{physics}
\usepackage{amsmath}
\usepackage{amssymb}
\usepackage{amsthm}
\usepackage{tikz}
\usepackage[breaklinks=true]{hyperref}

\theoremstyle{definition}

\newtheorem{theorem}{Theorem}[section]

\newtheorem*{exmp*}{Example}
\pgfdeclarelayer{background}
\pgfdeclarelayer{foreground}
\pgfsetlayers{background,main,foreground}

\begin{document}
\title{Gaussian optical networks for one-dimensional anyons}
\author{Allan D. C. Tosta}
\affiliation{Instituto de F\'isica, Universidade Federal Fluminense, Niter\'oi, RJ, 24210-340, Brazil}
\author{Ernesto F. Galv\~{a}o}
\affiliation{Instituto de F\'isica, Universidade Federal Fluminense, Niter\'oi, RJ, 24210-340, Brazil}
\affiliation{International Iberian Nanotechnology Laboratory (INL), Avenida Mestre José Veiga, 4715-330 Braga, Portugal}
\author{Daniel J. Brod}
\affiliation{Instituto de F\'isica, Universidade Federal Fluminense, Niter\'oi, RJ, 24210-340, Brazil}
\date{\today}

\begin{abstract}
We study the dynamics of bosonic and fermionic anyons defined on a one-dimensional lattice, under the effect of Hamiltonians quadratic in creation and annihilation operators, also called Gaussian Hamiltonians. These anyonic models are obtained from deformations of the standard bosonic or fermionic commutation relations via the introduction of a non-trivial exchange phase between different lattice sites. We study the effects of the anyonic exchange phase on the usual bosonic and fermionic bunching behaviors. We show how to exploit the inherent Aharonov-Bohm effect exhibited by these particles to build a deterministic, entangling two-qubit gate and prove quantum computational universality in these systems. We define coherent states for bosonic anyons and study their behavior under two-mode passive, Gaussian devices. In particular we prove that, for a particular value of the exchange factor, an anyonic mirror can generate cat states, an important resource in quantum information processing with continuous variables.
\end{abstract}

\maketitle

\section{Introduction}
Quantum computing models based on identical particle systems \cite{knill_scheme_2001,bravyi_fermionic_2002} have been extensively studied in the literature. In the case of bosonic systems, quadratic Hamiltonians (i.e., bosonic linear optics) are generated by beam-splitters over optical modes and, if supplemented with adaptive measurements, are universal for quantum computing \cite{knill_scheme_2001,kok_linear_2007}. For fermions, Hamiltonians of the same form (i.e., fermionic linear optics) cannot be used for universal quantum computing, since they can be simulated efficiently using classical algorithms \cite{terhal_classical_2002}. Classical simulability also holds for a class of quantum circuits called nearest-neighbour matchgates \cite{valiant_quantum_2002}, closely related to free fermions \cite{knill_fermionic_2001}, which has been interpreted as a kind of linear optics for qubits \cite{wu_qubits_2002}.

Over the last five decades, it was shown that fermions and bosons are not the only possible kinds of identical particles in Nature. Many physical systems in two dimensions were shown to contain quasi-particle excitations with anyonic statistics \cite{wilczek_magnetic_1982,nayak_non-abelian_2008}, where wave-functions acquire a non-trivial multiplicative phase under particle exchange. The most striking examples of this are fractional quantum Hall states \cite{stern_anyons_2008}, topological spin liquids \cite{savary_quantum_2016} and semiconductor nanowire arrays \cite{stanescu_majorana_2013}. These systems are possible platforms for fault-tolerant quantum computing \cite{nayak_non-abelian_2008,sarma_majorana_2015} and inspired new forms of quantum error-correcting codes \cite{fowler_surface_2012,brell_generalized_2015-1,litinski_quantum_2018}. 

Although anyons are most commonly associated to two-dimensional systems, they also exist in one dimension. Some arise as dimensional reduction of two-dimensional anyon states \cite{hansson_anyons_1991,ha_fractional_1995}, but most are obtained as free-particle descriptions of exactly solvable models with local two-body interactions in one dimension \cite{lieb_exact_1963-1,calogero_ground_1969,haldane_model_1988,shastry_exact_1988,olshanetsky_quantum_1983}.
The systems we consider are anyons defined via deformed commutation relations \cite{bozejko_completely_1994-1,jorgensen_q-canonical_1994,jorgensen_positive_1995,amico_one-dimensional_1998,osterloh_bethe_2000,osterloh_exact_2000,osterloh_fermionic_2000,osterloh_integrable_2001}, where the $\pm1$ bosonic/fermionic exchange phase is replaced by a non-trivial, mode-dependent complex phase $\pm e^{i\varphi}$, where the statistical parameter  $\varphi$ is between $0$ and $\pi$. When $\varphi$ is zero, the deformed commutation relations revert back to the canonical form, and thus $\varphi$ is a measure of the statistical deformation.

These types of anyons play a role in one-dimensional many-body systems with three-body interactions \cite{kundu_exact_1999,batchelor_one-dimensional_2006,batchelor_fermionization_2006,calabrese_correlation_2007,patu_correlation_2007,hao_ground-state_2008,keilmann_statistically_2011}, which have been subject to investigations in optical lattice implementations \cite{greschner_density-dependent_2014,cardarelli_engineering_2016,liang_floquet_2018,schweizer_floquet_2019}. They are not known to naturally occur, or to possess any special property with regards to fault-tolerance, so it is unclear whether they are suitable as platforms for quantum computing. Nevertheless, their relation to standard fermionic and bosonic systems via generalized Jordan-Wigner transformations \cite{meljanac_r_1994,doresic_generalized_1994,meljanac_unified_1996} makes them a good case study for generalizations of quantum computing with bosonic and fermionic linear optics.

By  exploiting similarities between the anyonic and the canonical commutation relations, we study the dynamics generated by Hamiltonians quadratic in anyonic creation and annihilation operators  (i.e.\ Gaussian devices). We use these dynamics to define anyonic optical elements and networks as generalization of bosonic and fermionic linear optics, and study the effects of the non-trivial exchange phase on their computational complexity. Guided by the phenomenology of quantum optics, we define anyonic coherent states and study their evolution under  anyonic optical elements on them. In particular, we show a simple protocol for generating anyonic cat states \cite{ralph_quantum_2003}.

We show that, in contrast to standard bosons and fermions, optical networks of anyonic Gaussian devices are universal for quantum computing for all values of $\varphi\neq0$. This is proved by constructing an entangling two-qubit gate from a network that exploits the non-trivial exchange phase of the anyons and preserves the encoding of logical qubits. This implies that introducing non-local statistical interactions in either bosonic or fermionic linear optics is equivalent, in terms of computing power, to introducing local, non-linear interactions \cite{kok_linear_2007}. This could have implications to the efficiency of classical simulations of one-dimensional anyon systems, as well as applications to quantum simulation algorithms in general.

This paper is structured as follows. In Section \ref{sec:I}, we give a brief review of bosonic and fermionic linear optics. In Section \ref{sec:II} we review the definition of bosonic and fermionic anyons. We discuss the anyonic generalization of optical networks and solve their dynamics. In Section \ref{sec:III} we show how to use optical networks of bosonic and fermionic anyons to build a protocol for universal quantum computing. Finally, in Section \ref{sec:IV} we give a brief introduction to the theory of optical coherence and the definition of generalized coherent states, investigating their analogue for bosonic anyons.

\section{Linear-optical networks}\label{sec:I}
In this Section we briefly review the theory of bosonic and fermionic linear optics \cite{divincenzo_fermionic_2005,kok_linear_2007}. In subsection \ref{sec:I.1} we review the second quantization formalism for identical particles and the definition of linear optics. In subsection \ref{sec:I.2}, we review some examples of the evolution of multiparticle states in multi-mode linear-optical devices. 

\subsection{Review of linear optics}\label{sec:I.1}
Consider a single spinless non-relativistic quantum particle on a one-dimensional lattice with open boundary conditions. Its Hilbert space is spanned by a discrete position basis, where we identify each position with a different mode. We denote this Hilbert space $\mathcal{H}_{m}$, where $m$ is the number of modes. From this space, the second quantization formalism gives a procedure for generating the Fock space for a system of such particles.

For bosons, the Fock space basis is built by acting with a set of operators $\{\hat{b}^{\dagger}_{i},i=1,...,m\}$ on a reference vacuum state $\ket{0}_b$. These operators satisfy the following commutation relations
\begin{subequations}\label{eq1}
\begin{align}
&\hat{b}_{i}\hat{b}^{\dagger}_{j}-\hat{b}^{\dagger}_{j}\hat{b}_{i}=\delta_{ij},\label{eq1a}\\
&\hat{b}_{i}\hat{b}_{j}-\hat{b}_{j}\hat{b}_{i}=0,\label{eq1b}\\
&\hat{b}^{\dagger}_{i}\hat{b}^{\dagger}_{j}-\hat{b}^{\dagger}_{j}\hat{b}^{\dagger}_{i}=0\label{eq1c},
\end{align}
\end{subequations}
for all pairs of modes $i$ and $j$. The basis vectors of the bosonic Fock space are given by
\begin{equation*}
\ket{{n}_{1},...,{n}_{m}}_{b}=\frac{(\hat{b}^{\dagger}_{1})^{{n}_{1}}...(\hat{b}^{\dagger}_{m})^{{n}_{m}}}{\sqrt{{n}_{1}!...{n}_{m}!}}\ket{0}_b,   
\end{equation*}
where ${n}_{i}$ is the eigenvalue of the number operator $\hat{n}_{i}=\hat{b}^{\dagger}_{i}\hat{b}_{i}$. This is also called the occupation number basis. 

From Eqs.\ (\ref{eq1}) follows the action of the bosonic operators on the occupation number basis:
\begin{align*}
\hat{b}_{i}\ket{{n}_{1},...,{n}_{i},...,{n}_{m}}_{b}&=\sqrt{{n}_{i}}\ket{{n}_{1},...,{n}_{i}-1,...,{n}_{m}}_{b},\\
\hat{b}^{\dagger}_{i}\ket{{n}_{1},...,{n}_{i},...,{n}_{m}}_{b}&=\sqrt{{n}_{i}+1}\ket{{n}_{1},...,{n}_{i}+1,...,{n}_{m}}_{b}.
\end{align*}
These expressions explain why $\hat{b}_{i}$ and $\hat{b}^{\dagger}_{i}$ are known as particle annihilation and creation operators, respectively.

For fermions, the Fock space is generated by a set of operators $\hat{f}^{\dagger}_{i}$, $\hat{f}_{i}$ that satisfy the anti-commutation relations
\begin{subequations}\label{eq6}
\begin{align}
&\hat{f}_{i}\hat{f}^{\dagger}_{j}+\hat{f}^{\dagger}_{j}\hat{f}_{i}=\delta_{ij},\label{eq6a}\\    
&\hat{f}_{i}\hat{f}_{j}+\hat{f}_{j}\hat{f}_{i}=0,\label{eq6b}\\
&\hat{f}^{\dagger}_{i}\hat{f}^{\dagger}_{j}+\hat{f}^{\dagger}_{j}\hat{f}^{\dagger}_{i}=0,\label{eq6c}
\end{align}
\end{subequations}
for all pairs of modes $i,j$. The number operator for mode $i$ is $\hat{n}_{i}=\hat{f}^{\dagger}_{i}\hat{f}_{i}$ and the vacuum state is $\ket{0}_{f}$. The Fock basis for fermions comprises the states
\begin{equation}
\ket{{n}_{1},...,{n}_{m}}_{f}=(\hat{f}^{\dagger}_{1})^{{n}_{1}}...(\hat{f}^{\dagger}_{m})^{{n}_{m}}\ket{0}_f, \label{eq:fermionFock}
\end{equation} 
where ${n}_{j}$ are eigenvalues of the corresponding number operators. From the anti-commutation relations it follows that the fermionic occupation numbers $n_{i}$ can only be $0$ or $1$.

For  both fermions and bosons, all observables can be written in terms of the particle operators.  Here we are interested in Hamiltonians of the form
\begin{equation}\label{eq_hamiltonian}
H=\sum_{i}a_{i}\hat{x}^{\dagger}_{i}\hat{x}_{i}+\sum_{i\neq j}{b}_{ij}\hat{x}^{\dagger}_{i}\hat{x}_{j}+\sum_{i\neq j}\left(c_{ij}\hat{x}^{\dagger}_{i}\hat{x}^{\dagger}_{j}+c^{*}_{ij}\hat{x}_{j}\hat{x}_{i}\right).
\end{equation}
Operators denoted as $\hat{x}$ correspond to either bosonic or fermionic operators, and we use this notation for expressions valid for both types of particles. Since  $H$ should be Hermitian, the coefficients $a_{i}$ are all real and ${b}_{ij}={b}^{*}_{ji}$.

The set of all Hamiltonians of the type above is closed under commutation and linear combinations, which gives it a Lie algebra structure. A convenient choice of generators is given by the quadratic operators $\hat{x}^{\dagger}_{i}\hat{x}_{j}$, $\hat{x}^{\dagger}_{i}\hat{x}^{\dagger}_{j}$ and $\hat{x}_{i}\hat{x}_{j}$, for all pairs $i,j$. In this work we focus on the sub-algebra of Hamiltonians that preserve the total number of particles. Such Hamiltonians are called number-preserving or passive, and correspond to setting $c_{ij}=0$ for all $i,j$ in Eq.\ (\ref{eq_hamiltonian}). From here on, we assume all Hamiltonians are passive unless stated otherwise.

From the commutation and anti-commutation relations it follows that the evolution operator $\hat{U}=\exp{i\theta H}$ corresponding to a passive Hamiltonian $H$ must, for any $\theta$, act on particle operators as
\begin{equation}\label{eq-7}
\hat{U}\hat{x}^{\dagger}_{i}\hat{U}^{\dagger}=\sum_{j}U_{i,j}\hat{x}^{\dagger}_{j},
\end{equation}
with a unitary matrix of coefficients $U=[U_{i,j}]$. Transformations of this type are known as passive linear dynamical maps, as they take creation operators into linear combinations of themselves. We refer to them as linear dynamics or linear optics for short. An arbitrary linear dynamics is uniquely determined by the $m\cross m$ matrix $U$ in Eq.\ (\ref{eq-7}). This definition is very useful, as we will see, since it allows us to obtain an evolved Fock state by direct substitution of the evolved particle operators.

In \cite{reck_experimental_1994}, it was shown that any $m\cross m$ unitary matrix can be decomposed as a product of $m(m-1)/2$ elementary transformations acting nontrivially only on two modes each. These elementary transformations can be implemented by combinations of simple optical devices known as phase shifters and beam splitters, which in turn can be represented as
\begin{align*}
{PS}_{i}(\tau)&=\exp\left(i\tau \hat{x}^{\dagger}_{i}\hat{x}_{i} \right),\\
{BS}_{ij}(\theta)&=\exp\left[i\theta(\hat{x}^{\dagger}_{i}\hat{x}_{j}+\hat{x}^{\dagger}_{j}\hat{x}_{i})\right].
\end{align*}
Their action on particle operators is given by
\begin{subequations}\label{eq-9}
\begin{align}
{PS}_{i}(\tau){x}^{\dagger}_{j}{PS}_{i}(-\tau)&=e^{i\tau\delta_{ij}}{x}^{\dagger}_{j}\label{eq-9a}\\
{BS}_{ij}(\theta)
\left[
\begin{array}{c}
{x}^{\dagger}_{i}\\
{x}^{\dagger}_{j}
\end{array}
\right]
{BS}_{ij}(-\theta)&=
\left[
\begin{array}{cc}
\cos{\theta} & i\sin{\theta}\\
i\sin{\theta} & \cos{\theta}
\end{array}
\right]
\left[
\begin{array}{c}
{x}^{\dagger}_{i}\\
{x}^{\dagger}_{j}
\end{array}
\right]
.
\end{align}\label{eq-9b}
\end{subequations}
In other words, the phase shifter adds a phase onto one mode relative to the others, and the beam splitter can place one particle in a superposition of being in different modes. The nomenclature for these devices is borrowed from optics, but they are well-defined transformations for any type of identical particle. This is analogous to what happens when an incident particle strikes a potential barrier, where its wave-function decomposes into a combination of transmitted and reflected waves.

Phase shifters and beam splitters can be combined into larger linear-optical networks and generate all transformations of the type of Eq.\ (\ref{eq-7}). We call these networks linear multi-mode interferometers, or simply interferometers for short. 

So far our description was completely equivalent for both bosons and fermions, but the action of these interferometers over multi-particle states of each type lead to many differences, which we explore next.

\subsection{Two-mode transformations and bunching behavior}\label{sec:I.2}
Consider a scenario where two particles must be sent through a 50:50 beam splitter [i.e.,  $\theta=\pi/4$ in Eq.\ (\ref{eq-9b}b)], for either fermions or bosons. The initial two-particle states are written as $\ket{1,1}_{f}=\hat{f}^{\dagger}_{1}\hat{f}^{\dagger}_{2}\ket{0}_f$ and  $\ket{1,1}_{b}=\hat{b}^{\dagger}_{1}\hat{b}^{\dagger}_{2}\ket{0}_b$, respectively. We now want to compare the corresponding output states. From Eqs.\ (\ref{eq-7}) and (\ref{eq-9b}b) we can write
\begin{align*}
BS_{12}\left(\tfrac{\pi}{4}\right)\ket{1,1}_{x}&=BS_{12}\left(\tfrac{\pi}{4}\right)\hat{x}^{\dagger}_{1}\hat{x}^{\dagger}_{2}\ket{0}_x\\
&=\frac{1}{2}\left[i(\hat{x}^{\dagger}_{1})^{2}+\hat{x}^{\dagger}_{1}\hat{x}^{\dagger}_{2}-\hat{x}^{\dagger}_{2}\hat{x}^{\dagger}_{1}+i(\hat{x}^{\dagger}_{2})^{2}\right]\ket{0}_x.
\end{align*}
where, recall, $x$ can denote either fermionic or bosonic operators. If we replace $x$ by $f$ and use the anti-commutation relations, we obtain
\begin{equation*}
BS_{12}\left(\tfrac{\pi}{4}\right)\ket{1,1}_{f}=\ket{1,1}_{f}.
\end{equation*}
In other words, the beam splitter is effectively transparent to the state $\ket{1,1}_{f}$. This is a manifestation of the Pauli exclusion principle.

Similarly, exchanging $x$ by $b$ and using the bosonic commutation relations we obtain
\begin{equation*}
BS_{12}\left(\tfrac{\pi}{4}\right)\ket{1,1}_{b}=\frac{i}{\sqrt{2}}\left(\ket{2,0}_{b}+\ket{0,2}_{b}\right).
\end{equation*}
In constrast to the fermionic case, the two bosons always exit together, in one mode or the other. This is known as the Hong-Ou-Mandel effect \cite{hong_measurement_1987}. These two effects are particular cases of the more general tendency of fermions and bosons to display, respectively, anti-bunching and bunching behaviors.

The generalization of these results to a larger interferometer $U$ is well-known and can be found in \cite{scheel_permanents_2004,terhal_classical_2002}. Specifically, if we have $n$ bosons (fermions) their transition amplitudes are given by permanents (determinants) of particular $n \times n$ submatrices of $U$. Interestingly, determinants are easy matrix functions to compute, whereas permanents are presumed to be very hard. This underpins a major difference in the computational power of bosonic and fermionic linear optics: fermionic linear optics can be simulated efficiently on a classical computer \cite{terhal_classical_2002}, whereas there is evidence that a bosonic linear-optical device cannot \cite{aaronson_computational_2011}.  However, neither is expected to be universal for quantum computing unless supplemented with further resources.

\section{One-dimensional anyons and optical networks}\label{sec:II}
In this Section we define the anyons we consider. There are several ways of defining anyons in one-dimensional systems \cite{hansson_anyons_1991,ha_fractional_1995,amico_one-dimensional_1998,kundu_exact_1999}, but we are interested in the so-called bosonic and fermionic anyons \cite{keilmann_statistically_2011}. Both are described in second quantization formalism with a Fock space basis generated by creation and annihilation operators.

In Section \ref{sec:II.1} we review the definition of bosonic and fermionic anyons and how they relate to standard bosons and fermions via generalized Jordan-Wigner transformations. In Section \ref{sec:II.2} we study the algebra of quadratic number-preserving Hamiltonians for anyons. Finally in Section \ref{II.3}, we define phase shifters and beam splitters for anyons, and how they act on anyonic Fock states.

\subsection{One-dimensional anyons}\label{sec:II.1}
Given a 1D lattice with $m$ sites, we denote the creation and annihilation operators for bosonic anyons by $\hat{\beta}^{\dagger}_{i}$ and $\hat{\beta}_{i}$ with $i=1,...,m$. They satisfy the deformed canonical quantization relations
\begin{subequations}\label{eq-16}
\begin{align}
&\hat{\beta}_{i}\hat{\beta}^{\dagger}_{j}-e^{-i\varphi\epsilon_{i,j}}\hat{\beta}^{\dagger}_{j}\hat{\beta}_{i}=\delta_{ij},\label{eq-16a}\\
&\hat{\beta}_{i}\hat{\beta}_{j}-e^{i\varphi\epsilon_{i,j}}\hat{\beta}_{j}\hat{\beta}_{i}=0,\label{eq-16b}\\
&\hat{\beta}^{\dagger}_{i}\hat{\beta}^{\dagger}_{j}-e^{i\varphi\epsilon_{i,j}}\hat{\beta}^{\dagger}_{j}\hat{\beta}^{\dagger}_{i}=0\label{eq-16c},
\end{align}
\end{subequations}
where $\epsilon_{i,j}$ is the sign of $j-i$, or $0$ if $i=j$. Notice that when $i=j$, we re-obtain the same-site bosonic commutation relations in Eqs.\ (\ref{eq1}). Particles obeying such commutation relations have been defined since the 1990's \cite{bozejko_completely_1994-1,jorgensen_q-canonical_1994,jorgensen_positive_1995}. They are related to operators for standard bosons by a transformation called a generalized Jordan-Wigner map $J_{\varphi}$ \cite{doresic_generalized_1994,meljanac_r_1994}, given by
\begin{subequations}
\begin{align*}
\hat{b}^{\dagger}_{i}\overset{J_{\varphi}}{\rightarrow}\hat{\beta}^{\dagger}_{i}&=\exp\left(-i\varphi\sum^{i-1}_{k=1}\hat{b}^{\dagger}_{k}\hat{b}_{k}\right)\hat{b}^{\dagger}_{i}\\
\hat{b}_{i}\overset{J_{\varphi}}{\rightarrow}\hat{\beta}_{i}&=\exp\left(i\varphi\sum^{i-1}_{k=1}\hat{b}^{\dagger}_{k}\hat{b}_{k}\right)\hat{b}_{i}.
\end{align*}
\end{subequations}
From these equations, and using the fact that $J_{\varphi}$ is an algebra homomorphism, i.e.\ $J_{\varphi}(ab)=J_{\varphi}(a)J_{\varphi}(b)$, it follows that $J_{\varphi}(\hat{b}^{\dagger}_{i}\hat{b}_{i})=\hat{\beta}^{\dagger}_{i}\hat{\beta}_{i}$. 

From Eqs.\ (\ref{eq-16}), it turns out that the operator $\hat{\beta}^{\dagger}_{i}\hat{\beta}_{i}$ shares properties of the bosonic operators $\hat{n}_{i}=\hat{b}^{\dagger}_{i}\hat{b}_{i}$. Given a vacuum state $\ket{0}_{\beta}$, the Fock space basis for bosonic anyons has the form
\begin{equation}
\ket{{n}_{1},...,{n}_{m}}_{\beta}=\frac{(\hat{\beta}^{\dagger}_{1})^{{n}_{1}}...(\hat{\beta}^{\dagger}_{m})^{{n}_{m}}}{\sqrt{{n}_{1}!...{n}_{m}!}}\ket{0}_\beta,   
\end{equation}
as in the standard bosonic case, where now $n_{i}$ are the eigenvalues of $\hat{n}_{i}=\hat{\beta}^{\dagger}_{i}\hat{\beta}_{i}$.

The action of particle operators on Fock basis states, however, is quite different. The exchange factors from the commutation relations appear, leading to
\small
\begin{align*}
\hat{\beta}_{i}\ket{{n}_{1},...,{n}_{m}}_{\beta}&=e^{i\varphi_{i}}\sqrt{{n}_{i}}\ket{{n}_{1},...,{n}_{i}-1,...,{n}_{m}}_{\beta},\\
\hat{\beta}^{\dagger}_{i}\ket{{n}_{1},...,{n}_{m}}_{\beta}&=e^{-i\varphi_{i}}\sqrt{{n}_{i}+1}\ket{{n}_{1},...,{n}_{i}+1,...,{n}_{m}}_{\beta},
\end{align*}
\normalsize
where the total phase $\varphi_{i}$ is equal to $\varphi\sum^{i-1}_{k=1}n_{k}$.

Fermionic anyons are defined in a similar way. We use the notation $\hat{\xi}^{\dagger}_{i}$ and $\hat{\xi}_{i}$ to refer to their creation and annihilation operators. The deformed anti-commutation relations are
\begin{subequations}\label{eq-20}
\begin{align}
\hat{\xi}_{i}\hat{\xi}^{\dagger}_{j}&+e^{-i\varphi\epsilon_{i,j}}\hat{\xi}^{\dagger}_{j}\hat{\xi}_{i}=\delta_{ij},\label{eq-20a}\\    
\hat{\xi}_{i}\hat{\xi}_{j}&+e^{i\varphi\epsilon_{i,j}}\hat{\xi}_{j}\hat{\xi}_{i}=0,\label{eq-20b}\\
\hat{\xi}^{\dagger}_{i}\hat{\xi}^{\dagger}_{j}&+e^{i\varphi\epsilon_{i,j}}\hat{\xi}^{\dagger}_{j}\hat{\xi}^{\dagger}_{i}=0.\label{eq-20c}
\end{align}
\end{subequations}
As expected, when we take $i=j$ we recover the single-mode canonical anti-commutation relations. 

As in the case of bosonic anyons, the generalized Jordan-Wigner transform applied on standard fermionic operators gives us
\begin{subequations}
\begin{align*}
\hat{f}^{\dagger}_{i}\overset{J_{\varphi}}{\rightarrow}\hat{\xi}^{\dagger}_{i}&=\exp\left(-i\varphi\sum^{i-1}_{k=1}\hat{f}^{\dagger}_{k}\hat{f}_{k}\right)\hat{f}^{\dagger}_{i}\\
\hat{f}_{i}\overset{J_{\varphi}}{\rightarrow}\hat{\xi}^{i}&=\exp\left(i\varphi\sum^{i-1}_{k=1}\hat{f}^{\dagger}_{k}\hat{f}_{k}\right)\hat{f}_{i},
\end{align*}
\end{subequations}
and we also have $J_{\varphi}(\hat{f}^{\dagger}_{i}\hat{f}_{i})=\hat{\xi}^{\dagger}_{i}\hat{\xi}_{i}$. This means that the Fock space basis states have the same structure as in the fermionic case, simply exchanging $f$ by $\xi$ in Eq.\ (\ref{eq:fermionFock}).

\subsection{The algebra of quadratic anyonic operators}\label{sec:II.2}
We now consider the algebra of quadratic Hamiltonians for anyons. As before, whenever we have an expression valid for both bosonic and fermionic anyons we express the particle operators using $\chi$, which can then be replaced $\beta$ or $\xi$, as appropriate.

Recall that all observables can be written in terms of the particle operators, and we can define passive, Gaussian Hamiltonians as
\begin{equation}\label{eq_hamiltonian_anyon}
H=\sum_{i}a_{i}\hat{\chi}^{\dagger}_{i}\hat{\chi}_{i}+\sum_{i\neq j}b_{ij}\hat{\chi}^{\dagger}_{i}\hat{\chi}_{j}.
\end{equation}
As before, the coefficients $\{a_{i},b_{ij}\}$ are such that $a_{i}$ are all real and $b_{ij}=b^{*}_{j,i}$, so that $H$ is Hermitian.

The first difference to the standard bosonic and fermionic Hamiltonians arises when we consider the closure of these operators. Though they form a closed set as a vector space, they are not closed under commutation. To see that, compare the commutators of quadratic operators for standard ($x$) and anyonic ($\chi$) particles:
\begin{align*}
[\hat{x}^{\dagger}_{i}\hat{x}_{j};\hat{x}^{\dagger}_{k}\hat{x}_{l}]&=\delta_{j,k}\hat{x}^{\dagger}_{i}\hat{x}_{l}-\delta_{i,l}\hat{x}^{\dagger}_{k}\hat{x}_{j}\\
[\hat{\chi}^{\dagger}_{i}\hat{\chi}_{j};\hat{\chi}^{\dagger}_{k}\hat{\chi}_{l}]&=\delta_{j,k}\hat{\chi}^{\dagger}_{i}\hat{\chi}_{l}-\delta_{i,l}\hat{\chi}^{\dagger}_{k}\hat{\chi}_{j}+\Delta^{\chi}_{i,j,k,l}\hat{\chi}^{\dagger}_{i}\hat{\chi}^{\dagger}_{j}\hat{\chi}_{k}\hat{\chi}_{l},
\end{align*}
where $\Delta^{\chi}_{i,j,k,l}$ is given by
\begin{equation*}
\Delta^{\chi}_{i,j,k,l}=
\begin{cases}
e^{-i\varphi\epsilon_{j,k}}-e^{-i\varphi(\epsilon_{l,i}-\epsilon_{k,i}-\epsilon_{l,j})},\text{ if }\chi=\beta,\\
-e^{-i\varphi\epsilon_{j,k}}+e^{-i\varphi(\epsilon_{l,i}-\epsilon_{k,i}-\epsilon_{l,j})},\text{ if }\chi=\xi.
\end{cases}
\end{equation*}
In other words, differently from standard bosons and fermions, the commutator of quadratic anyonic operators is not itself a quadratic operator, and so this algebra is not closed. Several consequences follow from this fact. The most important is that quadratic anyonic Hamiltonians are non-linear dynamical maps by default. In fact, the closure of the algebra of passive, Gaussian anyonic Hamiltonians must include number-preserving interaction terms of all even orders. 

Let us now restrict the algebra to Hamiltonians acting only on two modes. Define the operators
\begin{subequations}
\begin{align}
&J^{1}_{ij}=\frac{1}{2}\left(\hat{\chi}^{\dagger}_{i}\hat{\chi}_{j}+\hat{\chi}^{\dagger}_{j}\hat{\chi}_{i}\right), \label{eq:su2J1}\\
&J^{2}_{ij}=-\frac{i}{2}\left(\hat{\chi}^{\dagger}_{i}\hat{\chi}_{j}-\hat{\chi}^{\dagger}_{j}\hat{\chi}_{i}\right),\\
&J^{3}_{ij}=\frac{1}{2}\left(\hat{\chi}^{\dagger}_{i}\hat{\chi}_{i}-\hat{\chi}^{\dagger}_{j}\hat{\chi}_{j}\right),
\end{align}
\end{subequations}
where $i,j$ are fixed indices. Then we have that 
\begin{equation} \label{eq:su2}
[J^{k}_{ij};J^{l}_{ij}]=i\epsilon_{klm}J^{m}_{ij},
\end{equation}
for all $k,l,m=1,...,3$, as proven in Appendix \ref{app:su2}. 

We conclude that it is not possible to talk about anyonic inteferometers, in general, using the algebra of passive, Gaussian Hamiltonians. Nonetheless, we can model the action dynamics of \emph{networks} of two-mode anyonic interferometers (the analogues of beam splitters and phase shifters), where the dynamics acting on each mode pair can be modeled by an $SU(2)$ algebra. This observation forms the core of this work.

\subsection{Optical networks for anyons}\label{II.3}
Let us now describe general optical networks for bosonic and fermionic anyons, and investigate how their behaviours compare to standard bosons and fermions.

Since the algebra of two-mode Gaussian Hamiltonians is the same for standard and anyonic particles, it makes sense to define the optical elements for anyons analogously to those of fermions and bosons. Thus we have the anyonic phase shifters and beam splitters described respectively by
\begin{subequations}
\begin{align}
{PS}_{i}(\tau)&=\exp\left(i\tau\hat{\chi}^{\dagger}_{i}\hat{\chi}_{i}\right),\\
{BS}_{ij}(\theta)=&\exp\left[i\theta(\hat{\chi}^{\dagger}_{i}\hat{\chi}_{j}+\hat{\chi}^{\dagger}_{j}\hat{\chi}_{i})\right].\label{eq:anyonbs}
\end{align}
\end{subequations}

For standard particles, the result of Reck \textit{et al.} \cite{reck_experimental_1994} shows how to build an arbitrary $m$-mode interferometer using a network of O$(m^2)$ beam splitters and phase shifters. However, as we discussed, the algebra of anyonic quadratic Hamiltonians does not close. Therefore, a simple result similar to \cite{reck_experimental_1994} cannot hold for these particles, and so we choose to \emph{define} multimode anyonic interferometers in terms of optical networks of two-modes directly, even if they do not correspond to dynamics generated by quadratic Hamiltonians.

To see the implications of this definition, we begin by analyzing the dynamics of particle operators under the two-mode optical elements. The action of phase shifters is simple, since it is equivalent to that of standard particles:
\begin{equation*}
{PS}_{i}(\tau)\hat{\chi}^{\dagger}_{j}{PS}_{i}(-\tau)=e^{i\tau\delta_{ij}}\hat{\chi}^{\dagger}_{j}.
\end{equation*}

The action of beam splitters, on the other hand, is trickier and requires a separate analysis for fermionic and bosonic anyons.

In the fermionic case, the Pauli exclusion principle, together with the $SU(2)$ algebra, implies that
\begin{equation*}
{BS}_{ij}(\theta)=1+i\sin{\theta}(\hat{\xi}^{\dagger}_{i}\hat{\xi}_{j}+\hat{\xi}^{\dagger}_{j}\hat{\xi}_{i})+(\cos{\theta}-1)(\hat{\xi}^{\dagger}_{i}\hat{\xi}_{j}+\hat{\xi}^{\dagger}_{j}\hat{\xi}_{i})^{2}.
\end{equation*}
The canonical relations then imply the following identities 
\begin{subequations}
\begin{align*}
&{BS}_{ij}(\theta)\hat{\xi}^{\dagger}_{i}{BS}_{ij}(-\theta)=\cos{\theta}\hat{\xi}^{\dagger}_{i}+i\sin{\theta}\hat{\xi}^{\dagger}_{j}e^{i\varphi\hat{\xi}^{\dagger}_{i}\hat{\xi}_{i}},\\
&{BS}_{ij}(\theta)\hat{\xi}^{\dagger}_{j}{BS}_{ij}(-\theta)=\cos{\theta}\hat{\xi}^{\dagger}_{j}+i\sin{\theta}\hat{\xi}^{\dagger}_{i}e^{-i\varphi\hat{\xi}^{\dagger}_{j}\hat{\xi}_{j}}.
\end{align*}
\end{subequations}
Notice that the beam splitter action is non-linear, as expected. This gives a simpler derivation of our previous result in \cite{tosta_quantum_2019}, where we obtained these identities by solving the Heisenberg equations of motion.

For bosonic anyons the derivation is more involved, especially because we cannot rely on the Pauli exclusion principle to limit the degree of polynomials that appear. To begin, note that the beam splitter Hamiltonian in Eq.\ (\ref{eq:anyonbs}) is equal to $2J^{1}_{ij}$ from Eq.\ (\ref{eq:su2J1}). From the canonical relations we obtain
\begin{subequations}
\begin{align*}
(2{J}^{1}_{ij})\hat{\beta}^{\dagger}_{i}=&\hat{\beta}^{\dagger}_{i}\left[2(\cos{\varphi}{J}^{1}_{ij}-\sin{\varphi}{J}^{2}_{ij})\right]+\hat{\beta}^{\dagger}_{j},\\
(2{J}^{1}_{ij})\hat{\beta}^{\dagger}_{j}=&\hat{\beta}^{\dagger}_{j}\left[2(\cos{\varphi}{J}^{1}_{ij}-\sin{\varphi}{J}^{2}_{ij})\right]+\hat{\beta}^{\dagger}_{i}.
\end{align*}
\end{subequations}
Repeatedly multiplying these equations by $2J^{1}_{ij}$ on the left, we can recursively exponentiate $2J^{1}_{ij}$. This leads to the \emph{propagation identities} (proven in appendix \ref{app:propagation})
\begin{subequations}
\begin{align}\label{eq-32a}
{\hat{G}}^{n\varphi}_{ij}(\theta)\hat{\beta}^{\dagger}_{i}=&(\cos{\theta}\hat{\beta}^{\dagger}_{i}+ie^{-in\varphi}\sin{\theta}\hat{\beta}^{\dagger}_{j})\hat{G}^{(n+1)\varphi}_{ij}(\theta),\\
\label{eq-32b}{\hat{G}}^{n\varphi}_{ij}(\theta)\hat{\beta}^{\dagger}_{j}=&(\cos{\theta}\hat{\beta}^{\dagger}_{j}+ie^{in\varphi}\sin{\theta}\hat{\beta}^{\dagger}_{i})\hat{G}^{(n+1)\varphi}_{ij}(\theta),
\end{align}
\end{subequations}
where the operator ${\hat{G}}^{n\varphi}_{ij}(\theta)$ is given by
\begin{equation*}
{\hat{G}}^{n\varphi}_{ij}(\theta)=e^{in\varphi{J}^{3}_{ij}}BS_{ij}(\theta)e^{-in\varphi{J}^{3}_{ij}}.
\end{equation*}

From these identities, one can propagate the action of a beam splitter through any polynomial of particle operators.

Up to now we have investigated only the action of beam splitter $BS_{ij}$ on the operators of modes $i$ and $j$. For standard particles, this is the only non-trivial dynamics, but for anyonic particles the situation is quite different.

If $i$ and $j$ are not nearest neighbours, there exists at least one mode $k$ in between $i$ and $j$.  The canonical relations for both fermionic and bosonic anyons imply that
\begin{equation}\label{bs_aharonov_mode}
\hat{G}^{n\varphi}_{ij}(\theta)\hat{\chi}^{\dagger}_{k}=\hat{\chi}^{\dagger}_{k}\hat{G}^{(n+2)\varphi}_{ij}(\theta).
\end{equation}
Here we used the general definition of beam splitters and phase shifters to extend the definition of $\hat{G}^{n\varphi}_{ij}$ to the fermionic case. 

We have now given all ingredients to compute the dynamics of bosonic and fermionic anyon states under networks of optical elements. Let us discuss a few examples to illustrate.

Consider first a three mode lattice where we input one bosonic anyon in the first mode and vacuum in the third, and  a general beam splitter acting between modes 1 and 3. Consider that there might or not be a second particle in mode 2. We write the input state as $\ket{1,n,0}_{\beta}$, where  $n\in\{0,1\}$. The action of the beam splitter on this input state is given by
\begin{equation*}
BS_{13}(\theta)\ket{1,n,0}_{\beta}
=\cos{\theta}\ket{1,n,0}_{\beta}+ie^{-i n\varphi}\sin{\theta}\ket{0,n,1}_{\beta},
\end{equation*}
which follows from the repeated application of Eq. (\ref{bs_aharonov_mode}). Similarly, for fermionic anyons initialized on the state $\ket{1,n,0}_{\xi}$ we have that
\begin{equation*}
BS_{13}(\theta)\ket{1,n,0}_{\xi}=\cos{\theta}\ket{1,n,0}_{\xi}+i e^{-i n(\varphi+\pi)}\sin{\theta}\ket{0,n,1}_{\xi}.
\end{equation*}

The above example shows that, when acting with beam splitters between distant modes, the presence of anyons on intermediate modes induces an additional relative phase.  This is analogous to what is expected from the anyonic Aharonov-Bohm effect in two-dimensions \cite{wilczek_magnetic_1982} and in one-dimensional rings of lattice anyons \cite{haug_aharonov-bohm_2019}. These phases are a manifestation of the nonlinear character of anyonic beam splitters, and shows that they are also effectively nonlocal, in the sense that they depend on the presence of particles in modes where they are not obviously acting. These effects have important consequences that we will explore shortly.

Let us now consider whether the exchange phase affects the bunching behaviors of bosons and fermions discussed in Section \ref{sec:I.2}. To that end, let us first compute the effect of a balanced beam splitter acting on the input state $\ket{1,1}_{\beta}$ of bosonic anyons. The balanced beam splitter operator is given by $e^{i(\pi/2)J_{12}^{1}}$, and so we have
\small
\begin{align*}
e^{i\tfrac{\pi}{2}J_{12}^{1}}\hat{\beta}^{\dagger}_{1}\hat{\beta}^{\dagger}_{2}\ket{0}_{\beta} 
&=\frac{1}{\sqrt{2}}\left(\hat{\beta}^{\dagger}_{1}+i \hat{\beta}^{\dagger}_{2}\right)\hat{G}^{\varphi}_{12}\left(\tfrac{\pi}{4}\right)\hat{\beta}^{\dagger}_{2}\ket{0}_{\beta} \notag\\
& =\frac{1}{2}\left(\hat{\beta}^{\dagger}_{1}+i\hat{\beta}^{\dagger}_{2}\right)\left(i{e}^{i\varphi}\hat{\beta}^{\dagger}_{1}+\hat{\beta}^{\dagger}_{2}\right)\hat{G}^{2\varphi}_{12}\left(\tfrac{\pi}{4}\right)\ket{0}_{\beta} \notag\\
&=\frac{i}{\sqrt{2}}(e^{i\varphi}\ket{2,0}_{\beta}+\ket{0,2}_{\beta}),
\end{align*}
\normalsize
where the last equality follows from the commutation relations and the fact that $\hat{G}^{2\varphi}_{12}(\theta)$ acts trivially on the vacuum state.
Note that this recovers the original bosonic Hong-Ou-Mandel effect when $\varphi=0$, as expected. Interestingly, however the $\ket{1,1}$ state is still suppressed for any value of the exchange phase. The only difference to the bosonic case is a relative phase between states $\ket{2,0}_\beta$ and $\ket{0,2}_\beta$. 

Similarly, for fermionic anyons, a general beam splitter acting on $\ket{1,1}_{\xi}$ produces the output
\small
\begin{align*}
BS_{12}(\theta) \hat{\xi}^{\dagger}_{1}\hat{\xi}^{\dagger}_{2}\ket{0}_{\xi}
& =\left(\cos^{2}{\theta}\hat{\xi}^{\dagger}_{1}\hat{\xi}^{\dagger}_{2}-\sin^{2}{\theta}\hat{\xi}^{\dagger}_{2}e^{i\varphi\hat{\xi}^{\dagger}_{1}\hat{\xi}_{1}}\hat{\xi}^{\dagger}_{2}e^{-i\varphi\hat{\xi}^{\dagger}_{2}\hat{\xi}_{2}}\right)\ket{0}_{\xi}\\
&=\left(\cos^{2}{\theta}\hat{\xi}^{\dagger}_{1}\hat{\xi}^{\dagger}_{2}-\sin^{2}{\theta}e^{i\varphi}\hat{\xi}^{\dagger}_{2}\hat{\xi}^{\dagger}_{1}\right)\ket{0}_{\xi}\\
&=\hat{\xi}^{\dagger}_{1}\hat{\xi}^{\dagger}_{2}\ket{0}_{\xi}.
\end{align*}
\normalsize
In other words, we observe the Pauli exclusion principle for fermionic anyons in the same way as for standard fermions.

This example shows that, at least at the level of a single beam splitter, the bunching behavior of bosonic and fermionic anyons matches that of their standard counterparts. We leave it as an open question whether this remains true for multimode interferometers (though there is evidence that the behavior does change in that case \cite{liu_asymmetric_2018}). This is also a manifestation of an important distinction between two-dimensional anyons, which in a sense interpolate between bosons and fermions, and those we consider, which form two separate classes of particles.

Though the anyonic exchange phase did not alter the HOM and Pauli exclusion effects, it does play a crucial role in applications for quantum computing, as we will see now.

\section{Anyonic quantum computing}\label{sec:III}
In \cite{knill_scheme_2001} it was shown that photonic linear optical networks can be universal for quantum computing. This protocol replaces nonlinearities by adaptive measurements, though the entangling gates it produces are not deterministic. In contrast, it was shown in  \cite{divincenzo_fermionic_2005} that, for fermionic linear optics, even adaptive measurements are not sufficient, and therefore some nonlinearity is mandatory for universal quantum computation. 
On the anyonic side, in previous work \cite{tosta_quantum_2019} we showed how to perform universal quantum computing using only optical networks for fermionic anyons, without the need for any extra resource. Here, we describe a protocol for universal quantum computation that improves on the result of \cite{tosta_quantum_2019} by virtue of (i) being simpler for fermionic anyons and (ii) working for bosonic anyons as well.

We begin in subsection \ref{sec:III.1} by defining a graphical representation for anyonic optical networks, and proving that they cannot be uniquely described by their action on single-particle states. In Section \ref{sec:III.2} we review the dual-rail encoding for optical quantum computation and show how to construct  an entangling two-qubit gate for all bosonic and fermionic anyons with nonzero statistical parameter $\varphi$. 

\subsection{Multimode anyonic interferometry and the one-dimensional analogue of the Aharonov-Bohm effect}\label{sec:III.1}

Let us begin by introducing optical network diagrams to represent how optical elements are arranged in a circuit. These diagrams are composed of two building blocks. They are
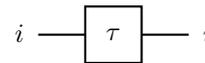
\begin{subequations}
\begin{equation*}
\begin{tikzpicture}[baseline,anchor=base]
\begin{scope}
\node[anchor=center] (c) at (0,0) {$\tau$};
\draw[thick] (0.375,0.375) -- (0.375,-0.375) -- (-0.375,-0.375) -- (-0.375,0.375) -- cycle;
\node (i1) at (-0.375,0) {};
\node (i2) at (0.375,0) {};
\node (j1) at (-1,0) {};
\node (j2) at (1,0) {};
\node (z1) at (-1.25,-0.1) {$i$};
\node (z2) at (1.25,-0.1) {$i$};
\draw[thick] (j1.center) to (i1.center);
\draw[thick] (j2.center) to (i2.center);
\end{scope}
\end{tikzpicture}
\end{equation*}
for a phase shifter $PS_{i}(\tau)$, and
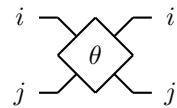
\begin{equation*}
\begin{tikzpicture}[baseline,anchor=base]
\node[anchor=center] (c) at (0,0) {$\theta$};
\node (s1-1) at (0,0.5) {};
\node (s2-1) at (0.5,0) {};
\node (s3-1) at (0,-0.5) {};
\node (s4-1) at (-0.5,0) {};
\draw[thick] (s1-1.center) to (s2-1.center);
\draw[thick] (s2-1.center) to (s3-1.center);
\draw[thick] (s3-1.center) to (s4-1.center);
\draw[thick] (s4-1.center) to (s1-1.center);
\node (i1) at (-0.25,0.25) {};
\node (i2) at (-0.25,-0.25) {};
\node (i3) at (0.25,0.25) {};
\node (i4) at (0.25,-0.25) {};
\node (j1) at (-0.5,0.5) {};
\node (j2) at (-0.5,-0.5) {};
\node (j3) at (0.5,0.5) {};
\node (j4) at (0.5,-0.5) {};
\draw[thick] (j1.center) to (i1.center);
\draw[thick] (j2.center) to (i2.center);
\draw[thick] (i3.center) to (j3.center);
\draw[thick] (i4.center) to (j4.center);
\node (e1) at (-0.75,0.5) {};
\node (e2) at (-0.75,-0.5) {};
\node (f1) at (0.75,0.5) {};
\node (f2) at (0.75,-0.5) {};
\draw[thick] (e1.center) to (j1.center);
\draw[thick] (e2.center) to (j2.center);
\draw[thick] (j3.center) to (f1.center);
\draw[thick] (j4.center) to (f2.center);

\node (x1) at (-1,0.4) {$i$};
\node (x2) at (1,0.4) {$i$};
\node (y1) at (-1,-0.6) {$j$};
\node (y2) at (1,-0.6) {$j$};
\end{tikzpicture}
\end{equation*}
for a beam splitter $BS_{ij}(\theta)$.
\end{subequations}
If the beam splitter acts between distant modes, the intermediate modes that go through it will be represented by dashed lines, as seen in Eq. (\ref{eq-42}).

In these diagrams, elements are composed from left to right, that is
\begin{equation*}
BS_{12}(\theta_{2})BS_{12}(\theta_{1})=
\begin{tikzpicture}[baseline,anchor=base]
\begin{scope}
\node[anchor=center] (c) at (0,0) {$\theta_{1}$};
\node (s1-1) at (0,0.5) {};
\node (s2-1) at (0.5,0) {};
\node (s3-1) at (0,-0.5) {};
\node (s4-1) at (-0.5,0) {};
\draw[thick] (s1-1.center) to (s2-1.center);
\draw[thick] (s2-1.center) to (s3-1.center);
\draw[thick] (s3-1.center) to (s4-1.center);
\draw[thick] (s4-1.center) to (s1-1.center);
\node (i1) at (-0.25,0.25) {};
\node (i2) at (-0.25,-0.25) {};
\node (i3) at (0.25,0.25) {};
\node (i4) at (0.25,-0.25) {};
\node (j1) at (-0.5,0.5) {};
\node (j2) at (-0.5,-0.5) {};
\node (j3) at (0.5,0.5) {};
\node (j4) at (0.5,-0.5) {};
\draw[thick] (j1.center) to (i1.center);
\draw[thick] (j2.center) to (i2.center);
\draw[thick] (i3.center) to (j3.center);
\draw[thick] (i4.center) to (j4.center);
\node (e1) at (-0.75,0.5) {};
\node (e2) at (-0.75,-0.5) {};
\node (f1) at (0.75,0.5) {};
\node (f2) at (0.75,-0.5) {};
\draw[thick] (e1.center) to (j1.center);
\draw[thick] (e2.center) to (j2.center);
\draw[thick] (j3.center) to (f1.center);
\draw[thick] (j4.center) to (f2.center);
\end{scope}
\begin{scope}[shift={(1.5,0)}]
\node[anchor=center] (c) at (0,0) {$\theta_{2}$};
\node (s1-1) at (0,0.5) {};
\node (s2-1) at (0.5,0) {};
\node (s3-1) at (0,-0.5) {};
\node (s4-1) at (-0.5,0) {};
\draw[thick] (s1-1.center) to (s2-1.center);
\draw[thick] (s2-1.center) to (s3-1.center);
\draw[thick] (s3-1.center) to (s4-1.center);
\draw[thick] (s4-1.center) to (s1-1.center);
\node (i1) at (-0.25,0.25) {};
\node (i2) at (-0.25,-0.25) {};
\node (i3) at (0.25,0.25) {};
\node (i4) at (0.25,-0.25) {};
\node (j1) at (-0.5,0.5) {};
\node (j2) at (-0.5,-0.5) {};
\node (j3) at (0.5,0.5) {};
\node (j4) at (0.5,-0.5) {};
\draw[thick] (j1.center) to (i1.center);
\draw[thick] (j2.center) to (i2.center);
\draw[thick] (i3.center) to (j3.center);
\draw[thick] (i4.center) to (j4.center);
\node (e1) at (-0.75,0.5) {};
\node (e2) at (-0.75,-0.5) {};
\node (f1) at (0.75,0.5) {};
\node (f2) at (0.75,-0.5) {};
\draw[thick] (e1.center) to (j1.center);
\draw[thick] (e2.center) to (j2.center);
\draw[thick] (j3.center) to (f1.center);
\draw[thick] (j4.center) to (f2.center);
\end{scope}

\node (x1) at (2.5,0.4) {1};
\node (y1) at (2.5,-0.6) {2};
\node (x2) at (-1,0.4) {1};
\node (y2) at (-1,-0.6) {2};
\end{tikzpicture}
\end{equation*}

As discussed previously, a linear-optical network for standard bosons or fermions implements a specified linear dynamical map $\hat{U}$ that is completely determined by its matrix elements $U=[U_{i,j}]$, i.e.\ by its action on single-particle subspace. Anyonic optical networks do not have this property. To prove this, consider the network below:
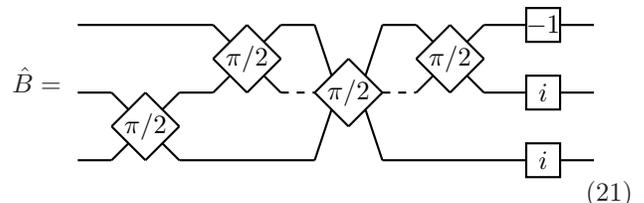
\begin{equation}\label{eq-42}
\hat{B}=
\begin{tikzpicture}[baseline,anchor=base,scale=0.9]
\begin{scope}[shift={(0.25,0)}]
\node (a1) at (-4,1) {};
\node (a2) at (-4,0) {};
\node (a3) at (-4,-1) {};
\node (bs23-2in) at (-3.5,0) {};
\node (bs23-3in) at (-3.5,-1) {};
\draw[thick] (a2.center) to (bs23-2in.center);
\draw[thick] (a3.center) to (bs23-3in.center);
\end{scope}
\begin{scope}[shift={(-2.75,-0.5)}]
\node[anchor=center] (c) at (0,0) {$\pi/2$};
\node (s1-1) at (0,0.5) {};
\node (s2-1) at (0.5,0) {};
\node (s3-1) at (0,-0.5) {};
\node (s4-1) at (-0.5,0) {};
\draw[thick] (s1-1.center) to (s2-1.center);
\draw[thick] (s2-1.center) to (s3-1.center);
\draw[thick] (s3-1.center) to (s4-1.center);
\draw[thick] (s4-1.center) to (s1-1.center);
\node (i1) at (-0.25,0.25) {};
\node (i2) at (-0.25,-0.25) {};
\node (i3) at (0.25,0.25) {};
\node (i4) at (0.25,-0.25) {};
\node (j1) at (-0.5,0.5) {};
\node (j2) at (-0.5,-0.5) {};
\node (j3) at (0.5,0.5) {};
\node (j4) at (0.5,-0.5) {};
\draw[thick] (j1.center) to (i1.center);
\draw[thick] (j2.center) to (i2.center);
\draw[thick] (i3.center) to (j3.center);
\draw[thick] (i4.center) to (j4.center);
\end{scope}
\begin{scope}[xscale=0.5,shift={(-2,0)}]
\node (bs23-2out) at (-2.5,0) {};
\node (bs23-3out) at (-2.5,-1) {};
\node (bs12-1in) at (-1.5,1) {};
\node (bs12-2in) at (-1.5,-0) {};
\draw[thick] (a1.center) to (bs12-1in.center);
\draw[thick] (bs23-2out.center) to (bs12-2in.center);
\end{scope}
\begin{scope}[shift={(-1.25,0.5)}]
\node[anchor=center] (c) at (0,0) {$\pi/2$};
\node (s1-1) at (0,0.5) {};
\node (s2-1) at (0.5,0) {};
\node (s3-1) at (0,-0.5) {};
\node (s4-1) at (-0.5,0) {};
\draw[thick] (s1-1.center) to (s2-1.center);
\draw[thick] (s2-1.center) to (s3-1.center);
\draw[thick] (s3-1.center) to (s4-1.center);
\draw[thick] (s4-1.center) to (s1-1.center);
\node (i1) at (-0.25,0.25) {};
\node (i2) at (-0.25,-0.25) {};
\node (i3) at (0.25,0.25) {};
\node (i4) at (0.25,-0.25) {};
\node (j1) at (-0.5,0.5) {};
\node (j2) at (-0.5,-0.5) {};
\node (j3) at (0.5,0.5) {};
\node (j4) at (0.5,-0.5) {};
\draw[thick] (j1.center) to (i1.center);
\draw[thick] (j2.center) to (i2.center);
\draw[thick] (i3.center) to (j3.center);
\draw[thick] (i4.center) to (j4.center);
\end{scope}
\begin{scope}[xscale=0.5,shift={(-1,0)}]
\node (bs12-1out) at (-0.5,1) {};
\node (bs12-2out) at (-0.5,-0) {};
\node (bs13-1in) at (0.5,1) {};
\node (bs13-2in) at (0.5,0) {};
\node (bs13-3in) at (0.5,-1) {};
\draw[thick] (bs12-1out.center) to (bs13-1in.center);
\draw[thick,dashed] (bs12-2out.center) to (bs13-2in.center);
\draw[thick] (bs23-3out.center) to (bs13-3in.center);
\end{scope}
\begin{scope}[shift={(0.25,0)}]
\node[anchor=center] (c) at (0,0) {$\pi/2$};
\node (s1-1) at (0,0.5) {};
\node (s2-1) at (0.5,0) {};
\node (s3-1) at (0,-0.5) {};
\node (s4-1) at (-0.5,0) {};
\draw[thick] (s1-1.center) to (s2-1.center);
\draw[thick] (s2-1.center) to (s3-1.center);
\draw[thick] (s3-1.center) to (s4-1.center);
\draw[thick] (s4-1.center) to (s1-1.center);
\node (i1) at (-0.25,0.25) {};
\node (i2) at (-0.25,-0.25) {};
\node (i3) at (0.25,0.25) {};
\node (i4) at (0.25,-0.25) {};
\node (j1) at (-0.5,1) {};
\node (j2) at (-0.5,-1) {};
\node (j3) at (0.5,1) {};
\node (j4) at (0.5,-1) {};
\draw[thick] (j1.center) to (i1.center);
\draw[thick] (j2.center) to (i2.center);
\draw[thick] (i3.center) to (j3.center);
\draw[thick] (i4.center) to (j4.center);
\end{scope}
\begin{scope}[xscale=0.5]
\node (bs13-1out) at (1.5,1) {};
\node (bs13-2out) at (1.5,0) {};
\node (bs12-1in2) at (2.5,1) {};
\node (bs12-2in2) at (2.5,-0) {};
\draw[thick] (bs13-1out.center) to (bs12-1in2.center);
\draw[thick,dashed] (bs13-2out.center) to (bs12-2in2.center);
\end{scope}
\begin{scope}[shift={(1.75,0.5)}]
\node[anchor=center] (c) at (0,0) {$\pi/2$};
\node (s1-1) at (0,0.5) {};
\node (s2-1) at (0.5,0) {};
\node (s3-1) at (0,-0.5) {};
\node (s4-1) at (-0.5,0) {};
\draw[thick] (s1-1.center) to (s2-1.center);
\draw[thick] (s2-1.center) to (s3-1.center);
\draw[thick] (s3-1.center) to (s4-1.center);
\draw[thick] (s4-1.center) to (s1-1.center);
\node (i1) at (-0.25,0.25) {};
\node (i2) at (-0.25,-0.25) {};
\node (i3) at (0.25,0.25) {};
\node (i4) at (0.25,-0.25) {};
\node (j1) at (-0.5,0.5) {};
\node (j2) at (-0.5,-0.5) {};
\node (j3) at (0.5,0.5) {};
\node (j4) at (0.5,-0.5) {};
\draw[thick] (j1.center) to (i1.center);
\draw[thick] (j2.center) to (i2.center);
\draw[thick] (i3.center) to (j3.center);
\draw[thick] (i4.center) to (j4.center);
\end{scope}
\begin{scope}[shift={(-1.5,0)}]
\node (p1) at (4.5,1) {};
\node (p2) at (4.5,0) {};
\node (p3) at (4.5,-1) {};
\end{scope}
\begin{scope}[shift={(0,0)},xscale=0.5]
\node (bs12-1out2) at (4.5,1) {};
\node (bs12-2out2) at (4.5,-0) {};
\node (bs13-3out) at (1.5,-1) {};
\draw[thick] (bs12-1out2.center) to (p1);
\draw[thick] (bs12-2out2.center) to (p2);
\draw[thick] (bs13-3out.center) to (p3);
\end{scope}
\begin{scope}[shift={(3.125,1)}]
\node[anchor=center] (c) at (0,0) {$-1$};
\draw[thick] (0.25,0.25) -- (0.25,-0.25) -- (-0.25,-0.25) -- (-0.25,0.25) -- cycle;
\node (i1) at (-0.25,0) {};
\node (i2) at (0.25,0) {};
\node (j1) at (-1,0) {};
\node (j2) at (0.75,0) {};
\draw[thick] (i2.center) to (j2.center);
\end{scope}
\begin{scope}[shift={(3.125,0)}]
\node[anchor=center] (c) at (0,0) {$i$};
\draw[thick] (0.25,0.25) -- (0.25,-0.25) -- (-0.25,-0.25) -- (-0.25,0.25) -- cycle;
\node (i1) at (-0.25,0) {};
\node (i2) at (0.25,0) {};
\node (j1) at (-1,0) {};
\node (j2) at (0.75,0) {};
\draw[thick] (i2.center) to (j2.center);
\end{scope}
\begin{scope}[shift={(3.125,-1)}]
\node[anchor=center] (c) at (0,0) {$i$};
\draw[thick] (0.25,0.25) -- (0.25,-0.25) -- (-0.25,-0.25) -- (-0.25,0.25) -- cycle;
\node (i1) at (-0.25,0) {};
\node (i2) at (0.25,0) {};
\node (j1) at (-1,0) {};
\node (j2) at (0.75,0) {};
\draw[thick] (i2.center) to (j2.center);
\end{scope}
\end{tikzpicture}    
\end{equation}
For the single particle subspace spanned by $\{\ket{1,0,0}_{\chi}, \ket{0,1,0}_{\chi}, \ket{0,0,1}_{\chi}\}$, the network acts as the identity for standard and anyonic particles of both types. For standard bosons and fermions, this implies that this network acts as the identity transformation on the whole Fock space.

In contrast, we can show that the action of this network on states with two or three anyons (either bosonic or fermionic) is nontrivial:
\begin{subequations}\label{eq-43}
\begin{align}
\hat{B}\ket{0,1,1}_{\chi}&=\ket{0,1,1}_{\chi},\\
\hat{B}\ket{1,0,1}_{\chi}&=e^{-i\varphi}\ket{1,0,1}_{\chi},\\
\hat{B}\ket{1,1,0}_{\chi}&=e^{i\varphi}\ket{1,1,0}_{\chi},\\
\hat{B}\ket{1,1,1}_{\chi}&=\ket{1,1,1}_{\chi}.
\end{align}
\end{subequations}
This network acts diagonally on this subspace, as expected, since each element only permutes the anyons between the modes but does not create superpositions. However, some diagonal elements are nontrivial and proportional to anyonic exchange phases. The existence of such nontrivial networks shows that the unique specification of interferometers based on their action on single-particle states is no longer possible.

\subsection{Dual-rail universal quantum computer}\label{sec:III.2}
Arguably the most common way to encode a qubit for use in linear-optical computing is the dual-rail encoding. In this encoding, $n$ qubits are mapped to the states of $n$ particles in $2n$ modes, such that each logical qubit is supported in a pair of neighboring modes. The logical qubit states are defined by
\begin{subequations}
\begin{align*}
\ket{0_{L}}&=\ket{1,0},\\
\ket{1_{L}}&=\ket{0,1}.
\end{align*}
\end{subequations}
Note that these states have at most a single particle per mode, and so they are supported by all types of particle we consider here. Therefore, we temporarily drop the reference to the particle type in the notation, unless where necessary.

To encode more qubits is straightforward. A two-qubit system needs four modes, with corresponding logical states
\begin{subequations}
\begin{align*}
\ket{00}_{L}&=\ket{1,0,1,0},\\
\ket{01}_{L}&=\ket{1,0,0,1},\\
\ket{10}_{L}&=\ket{0,1,1,0},\\
\ket{11}_{L}&=\ket{0,1,0,1}.
\end{align*}
\end{subequations}

With this encoding it is possible to perform any logical single-qubit gate using only phase shifters and beam splitters. To prove this, consider a  qubit encoded in modes 1 and 2. A phase shifter on mode 2 acts in the logical basis states as
\begin{subequations}
\begin{align*}
&{PS}_{2}(\theta)\ket{1,0}=\ket{1,0},\\
&{PS}_{2}(\theta)\ket{0,1}=e^{i\theta}\ket{0,1},
\end{align*}
\end{subequations}
which is a logical $Z$ rotation on the Bloch sphere by $\theta$. A beam splitter between modes 1 and 2 acts in the logical basis states as
\begin{subequations}
\begin{align*}
&{BS}_{12}(\theta)\ket{1,0}=\cos{\theta}\ket{1,0}+i\sin{\theta}\ket{0,1},\\
&{BS}_{12}(\theta)\ket{0,1}=i\sin{\theta}\ket{1,0}+\cos{\theta}\ket{0,1},
\end{align*}
\end{subequations}
which is a logical $X$ rotation in the Bloch sphere by an angle $\theta$. With arbitrary rotations around two distinct axes in the Bloch sphere we can perform arbitrary single-qubit gates \cite{nielsen_quantum_2010} via the decomposition
\begin{equation*}
U=e^{i\alpha}e^{-i\frac{\beta Z}{2}}e^{-i\frac{\gamma X}{2}}e^{-i\frac{\delta Z}{2}}.
\end{equation*}
An optical realization of this decomposition is given by the network
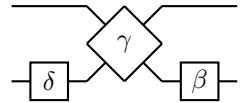
\begin{equation*}
\begin{tikzpicture}[baseline,anchor=base]
\begin{scope}[shift={(-1,-0.5)}]
\node[anchor=center] (c) at (0,0) {$\delta$};
\draw[thick] (0.25,0.25) -- (0.25,-0.25) -- (-0.25,-0.25) -- (-0.25,0.25) -- cycle;
\node (i1) at (-0.25,0) {};
\node (i2) at (0.25,0) {};
\node (j1) at (-0.5,0) {};
\node (j2) at (0.5,0) {};
\draw[thick] (j1.center) to (i1.center);
\end{scope}

\node[anchor=center] (c) at (0,0) {$\gamma$};
\node (s1-1) at (0,0.5) {};
\node (s2-1) at (0.5,0) {};
\node (s3-1) at (0,-0.5) {};
\node (s4-1) at (-0.5,0) {};
\draw[thick] (s1-1.center) to (s2-1.center);
\draw[thick] (s2-1.center) to (s3-1.center);
\draw[thick] (s3-1.center) to (s4-1.center);
\draw[thick] (s4-1.center) to (s1-1.center);
\node (i1) at (-0.25,0.25) {};
\node (i2) at (-0.25,-0.25) {};
\node (i3) at (0.25,0.25) {};
\node (i4) at (0.25,-0.25) {};
\node (j1) at (-0.5,0.5) {};
\node (j2) at (-0.5,-0.5) {};
\node (j3) at (0.5,0.5) {};
\node (j4) at (0.5,-0.5) {};
\draw[thick] (j1.center) to (i1.center);
\draw[thick] (j2.center) to (i2.center);
\draw[thick] (i3.center) to (j3.center);
\draw[thick] (i4.center) to (j4.center);
\node (e1) at (-1.5,0.5) {};
\node (e2) at (-0.75,-0.5) {};
\node (f1) at (1.5,0.5) {};
\node (f2) at (0.75,-0.5) {};
\draw[thick] (e1.center) to (j1.center);
\draw[thick] (e2.center) to (j2.center);
\draw[thick] (j3.center) to (f1.center);
\draw[thick] (j4.center) to (f2.center);

\begin{scope}[shift={(1,-0.5)}]
\node[anchor=center] (c) at (0,0) {$\beta$};
\draw[thick] (0.25,0.25) -- (0.25,-0.25) -- (-0.25,-0.25) -- (-0.25,0.25) -- cycle;
\node (i1) at (-0.25,0) {};
\node (i2) at (0.25,0) {};
\node (j1) at (-0.5,0) {};
\node (j2) at (0.5,0) {};
\draw[thick] (j2.center) to (i2.center);
\end{scope}
\end{tikzpicture}
\end{equation*}

To build a universal computer, we must also have an entangling two-qubit gate \cite{lloyd_almost_1995}. For standard particles, deterministic entangling gates require nonlinear interactions. For the anyonic particles we consider the nonlinearity is intrinsic, and so the task of finding an appropriate  gate is simpler.
Consider the optical network below
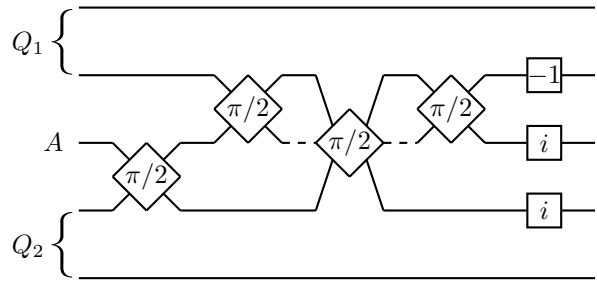
\begin{equation*}
\begin{tikzpicture}[baseline,anchor=base,scale=0.9]
\node (a4) at (-3.75,2) {};
\node (a5) at (-3.75,-2) {};

\node (c1) at (-4,1.25) {\Huge{\{}};
\node (c2) at (-4,-1.75) {\Huge{\{}};

\node (c1) at (-4.5,1.375) {$Q_{1}$};
\node (c2) at (-4.5,-1.625) {$Q_{2}$};

\node (x) at (-4.125,-0.125) {$A$};

\begin{scope}[shift={(0.25,0)}]
\node (a1) at (-4,1) {};
\node (a2) at (-4,0) {};
\node (a3) at (-4,-1) {};
\node (bs23-2in) at (-3.5,0) {};
\node (bs23-3in) at (-3.5,-1) {};
\draw[thick] (a2.center) to (bs23-2in.center);
\draw[thick] (a3.center) to (bs23-3in.center);
\end{scope}
\begin{scope}[shift={(-2.75,-0.5)}]
\node[anchor=center] (c) at (0,0) {$\pi/2$};
\node (s1-1) at (0,0.5) {};
\node (s2-1) at (0.5,0) {};
\node (s3-1) at (0,-0.5) {};
\node (s4-1) at (-0.5,0) {};
\draw[thick] (s1-1.center) to (s2-1.center);
\draw[thick] (s2-1.center) to (s3-1.center);
\draw[thick] (s3-1.center) to (s4-1.center);
\draw[thick] (s4-1.center) to (s1-1.center);
\node (i1) at (-0.25,0.25) {};
\node (i2) at (-0.25,-0.25) {};
\node (i3) at (0.25,0.25) {};
\node (i4) at (0.25,-0.25) {};
\node (j1) at (-0.5,0.5) {};
\node (j2) at (-0.5,-0.5) {};
\node (j3) at (0.5,0.5) {};
\node (j4) at (0.5,-0.5) {};
\draw[thick] (j1.center) to (i1.center);
\draw[thick] (j2.center) to (i2.center);
\draw[thick] (i3.center) to (j3.center);
\draw[thick] (i4.center) to (j4.center);
\end{scope}
\begin{scope}[xscale=0.5,shift={(-2,0)}]
\node (bs23-2out) at (-2.5,0) {};
\node (bs23-3out) at (-2.5,-1) {};
\node (bs12-1in) at (-1.5,1) {};
\node (bs12-2in) at (-1.5,-0) {};
\draw[thick] (a1.center) to (bs12-1in.center);
\draw[thick] (bs23-2out.center) to (bs12-2in.center);
\end{scope}
\begin{scope}[shift={(-1.25,0.5)}]
\node[anchor=center] (c) at (0,0) {$\pi/2$};
\node (s1-1) at (0,0.5) {};
\node (s2-1) at (0.5,0) {};
\node (s3-1) at (0,-0.5) {};
\node (s4-1) at (-0.5,0) {};
\draw[thick] (s1-1.center) to (s2-1.center);
\draw[thick] (s2-1.center) to (s3-1.center);
\draw[thick] (s3-1.center) to (s4-1.center);
\draw[thick] (s4-1.center) to (s1-1.center);
\node (i1) at (-0.25,0.25) {};
\node (i2) at (-0.25,-0.25) {};
\node (i3) at (0.25,0.25) {};
\node (i4) at (0.25,-0.25) {};
\node (j1) at (-0.5,0.5) {};
\node (j2) at (-0.5,-0.5) {};
\node (j3) at (0.5,0.5) {};
\node (j4) at (0.5,-0.5) {};
\draw[thick] (j1.center) to (i1.center);
\draw[thick] (j2.center) to (i2.center);
\draw[thick] (i3.center) to (j3.center);
\draw[thick] (i4.center) to (j4.center);
\end{scope}
\begin{scope}[xscale=0.5,shift={(-1,0)}]
\node (bs12-1out) at (-0.5,1) {};
\node (bs12-2out) at (-0.5,-0) {};
\node (bs13-1in) at (0.5,1) {};
\node (bs13-2in) at (0.5,0) {};
\node (bs13-3in) at (0.5,-1) {};
\draw[thick] (bs12-1out.center) to (bs13-1in.center);
\draw[thick,dashed] (bs12-2out.center) to (bs13-2in.center);
\draw[thick] (bs23-3out.center) to (bs13-3in.center);
\end{scope}
\begin{scope}[shift={(0.25,0)}]
\node[anchor=center] (c) at (0,0) {$\pi/2$};
\node (s1-1) at (0,0.5) {};
\node (s2-1) at (0.5,0) {};
\node (s3-1) at (0,-0.5) {};
\node (s4-1) at (-0.5,0) {};
\draw[thick] (s1-1.center) to (s2-1.center);
\draw[thick] (s2-1.center) to (s3-1.center);
\draw[thick] (s3-1.center) to (s4-1.center);
\draw[thick] (s4-1.center) to (s1-1.center);
\node (i1) at (-0.25,0.25) {};
\node (i2) at (-0.25,-0.25) {};
\node (i3) at (0.25,0.25) {};
\node (i4) at (0.25,-0.25) {};
\node (j1) at (-0.5,1) {};
\node (j2) at (-0.5,-1) {};
\node (j3) at (0.5,1) {};
\node (j4) at (0.5,-1) {};
\draw[thick] (j1.center) to (i1.center);
\draw[thick] (j2.center) to (i2.center);
\draw[thick] (i3.center) to (j3.center);
\draw[thick] (i4.center) to (j4.center);
\end{scope}
\begin{scope}[xscale=0.5]
\node (bs13-1out) at (1.5,1) {};
\node (bs13-2out) at (1.5,0) {};
\node (bs12-1in2) at (2.5,1) {};
\node (bs12-2in2) at (2.5,-0) {};
\draw[thick] (bs13-1out.center) to (bs12-1in2.center);
\draw[thick,dashed] (bs13-2out.center) to (bs12-2in2.center);
\end{scope}
\begin{scope}[shift={(1.75,0.5)}]
\node[anchor=center] (c) at (0,0) {$\pi/2$};
\node (s1-1) at (0,0.5) {};
\node (s2-1) at (0.5,0) {};
\node (s3-1) at (0,-0.5) {};
\node (s4-1) at (-0.5,0) {};
\draw[thick] (s1-1.center) to (s2-1.center);
\draw[thick] (s2-1.center) to (s3-1.center);
\draw[thick] (s3-1.center) to (s4-1.center);
\draw[thick] (s4-1.center) to (s1-1.center);
\node (i1) at (-0.25,0.25) {};
\node (i2) at (-0.25,-0.25) {};
\node (i3) at (0.25,0.25) {};
\node (i4) at (0.25,-0.25) {};
\node (j1) at (-0.5,0.5) {};
\node (j2) at (-0.5,-0.5) {};
\node (j3) at (0.5,0.5) {};
\node (j4) at (0.5,-0.5) {};
\draw[thick] (j1.center) to (i1.center);
\draw[thick] (j2.center) to (i2.center);
\draw[thick] (i3.center) to (j3.center);
\draw[thick] (i4.center) to (j4.center);
\end{scope}
\begin{scope}[shift={(-1.5,0)}]
\node (p1) at (4.5,1) {};
\node (p2) at (4.5,0) {};
\node (p3) at (4.5,-1) {};
\end{scope}
\begin{scope}[shift={(0,0)},xscale=0.5]
\node (bs12-1out2) at (4.5,1) {};
\node (bs12-2out2) at (4.5,-0) {};
\node (bs13-3out) at (1.5,-1) {};
\draw[thick] (bs12-1out2.center) to (p1);
\draw[thick] (bs12-2out2.center) to (p2);
\draw[thick] (bs13-3out.center) to (p3);
\end{scope}
\begin{scope}[shift={(3.125,1)}]
\node[anchor=center] (c) at (0,0) {$-1$};
\draw[thick] (0.25,0.25) -- (0.25,-0.25) -- (-0.25,-0.25) -- (-0.25,0.25) -- cycle;
\node (i1) at (-0.25,0) {};
\node (i2) at (0.25,0) {};
\node (j1) at (-1,0) {};
\node (j2) at (0.75,0) {};
\draw[thick] (i2.center) to (j2.center);
\end{scope}
\begin{scope}[shift={(3.125,0)}]
\node[anchor=center] (c) at (0,0) {$i$};
\draw[thick] (0.25,0.25) -- (0.25,-0.25) -- (-0.25,-0.25) -- (-0.25,0.25) -- cycle;
\node (i1) at (-0.25,0) {};
\node (i2) at (0.25,0) {};
\node (j1) at (-1,0) {};
\node (j2) at (0.75,0) {};
\draw[thick] (i2.center) to (j2.center);
\end{scope}
\begin{scope}[shift={(3.125,-1)}]
\node (c) at (0,0) {$i$};
\draw[thick] (0.25,0.25) -- (0.25,-0.25) -- (-0.25,-0.25) -- (-0.25,0.25) -- cycle;
\node (i1) at (-0.25,0) {};
\node (i2) at (0.25,0) {};
\node (j1) at (-1,0) {};
\node (j2) at (0.75,0) {};
\draw[thick] (i2.center) to (j2.center);
\end{scope}
\node (b4) at (3.875,2) {};
\node (b5) at (3.875,-2) {};

\draw[thick] (a4.center) to (b4.center);
\draw[thick] (a5.center) to (b5.center);
\end{tikzpicture}
\end{equation*}
where modes associated to $Q_{1}$ encode the first qubit and modes in $Q_{2}$ encode the second qubit. Mode $A$ is an auxiliary mode. This network is equivalent to applying the network in Eq.\ (\ref{eq-42}) to three modes in the middle, leaving the outermost two untouched. Therefore, if we initialize the auxiliary mode $A$ with one particle, Eqs.\ (\ref{eq-43}) show that this network generates the two-qubit gate
\begin{equation*}
CP(\varphi)=
\begin{bmatrix}
1 & 0 & 0 & 0\\
0 & 1 & 0 & 0\\
0 & 0 & 1 & 0\\
0 & 0 & 0 & e^{i\varphi}\\
\end{bmatrix}
\end{equation*}
This gate is a controlled phase gate, which is an entangling gate for all $\varphi\neq0$. 

In previous work \cite{tosta_quantum_2019} we showed that fermionic anyons for any $\varphi \neq 0$ are universal for quantum computation. The above construction proves quantum universality for both fermionic and bosonic anyons for any $\varphi\neq0$, using only optical networks and one auxiliary mode, with one particle that never leaves the circuit. Besides also holding for bosonic anyons, it improves over the construction of \cite{tosta_quantum_2019} for fermionic anyons by being simpler. 

This ends our applications of anyonic optical networks for quantum computing with discrete variables. In the next Section we  move to the realm of continuous variables.

\section{Coherent states of anyons}\label{sec:IV}
Bosonic coherent states are the platform of many proposals for quantum computing with continuous variables \cite{jeong_efficient_2002,ralph_quantum_2003}. In this Section, we define the anyonic analogue of bosonic coherent states.

In subsection \ref{sec:IV-1}, we review the general quantum theory of coherence as formulated in \cite{glauber_coherent_1963,glauber_quantum_1963,titulaer_density_1966,bialynickabirula_properties_1968,stoler_generalized_1971}, focusing on single-mode coherent states, and extend it to bosonic anyons. In subsection \ref{sec:IV-2} we show how optical elements act on anyonic coherent states.

\subsection{Quantum theory of coherence and anyons}\label{sec:IV-1}
Let us begin by considering single-mode coherent states (and dropping the mode index from all expressions), and defer the multimode case to the next subsection.

Coherent states are usually defined as eigenstates of annihilation operators. For example, for standard bosons, a coherent state satisfies
\begin{equation} \label{eq:coherent}
\hat{b}\ket{g}_{b}=g\ket{g}_{b}.
\end{equation}
Here, $g$, known as the amplitude of the coherent state, can be any complex number (due to the non-Hermiticity of $\hat{b}$). In the Fock basis, the coherent state can be written as
\begin{equation}
\ket{g}_{b}=e^{-\frac{1}{2}\abs{g}^{2}}\sum_{n}\frac{(g\hat{b}^{\dagger})^{n}}{n!}\ket{0}_{b},
\end{equation}

The states defined by Eq.\  (\ref{eq:coherent}), however, are only a particular kind of coherent state. The theory of optical quantum coherence arose from the task of discriminating experimentally between different states of the electromagnetic field  by the amplitude of $n$-photon absorption events \cite{glauber_coherent_1963,glauber_quantum_1963}. Given a single-mode input state $\ket{input}_{b}$ coming from some field source, the probability of detecting $n$ photons is given by the $n$-th order correlation function
\begin{equation*}
C_{b}(n)=\bra{input}{(\hat{b}^{\dagger})^{n}(\hat{b})^{n}}\ket{input}_{b}.
\end{equation*}
These correlators can be used, for example, as a measure to attest the quality of single-photon sources \cite{spring_boson_2013,crespi_integrated_2013,spagnolo_experimental_2014}, since for such sources we should have $C_{b}(1)$ as high as possible. It is more common, however, to use the so-called the $n$-th order single-mode coherence functions
\begin{equation}
c_{b}(n)=\frac{\expval{(\hat{b}^{\dagger})^{n}(\hat{b})^{n}}_{b}}{\expval{n}^{n}_{b}},
\end{equation}
where $c_{b}(n)$ is calculated relative to some specific state. We say that a state $\ket{\psi}_{b}$ is $n$-order coherent if $c_{b}(m)=1$ for all $m\leq n$. In the general theory of coherence, a \emph{coherent state}  is one for which $c_{b}(n)=1$ for all $n\in\mathbb{N}$. In other words, this state has full coherence in the sense that it is $n$-order coherent for all $n$.

The most general coherent state is of the form
\begin{equation}
\ket{g|\{\rho_{n}\}}_{b}=e^{-\frac{1}{2}\abs{g}^{2}}\sum_{n}e^{i\rho_{n}}\frac{(g\hat{\beta}^{\dagger})^{n}}{n!}\ket{0}_{b},
\end{equation}
where $\rho_{n}$ is an arbitrary sequence of real numbers \cite{titulaer_density_1966,bialynickabirula_properties_1968}. Note that we recover the states defined in Eq.\ (\ref{eq:coherent}) in the particular case of $\rho_n=0$ for all  $n$.
 
Let us now consider the analogue of coherent states for bosonic anyons. We may define eigenstates of annihilation operators:
\begin{equation}
\hat{\beta}\ket{g}_{\beta}=g\ket{g_{\beta}},
\end{equation}
where again $g$ is any complex number. These states  have full coherence if we define the $n$th order coherence $c_{\beta}(n)$ function in the analogous form
\begin{equation} \label{eq:coherences}
c_{\beta}(n)=\frac{\expval{(\hat{\beta}^{\dagger})^{n}(\hat{\beta})^{n}}_{\beta}}{\expval{n}^{n}_{\beta}}  
\end{equation}
We thus also refer to these as coherent states. 

Note that single-mode bosonic anyons satisfy the same commutation relations as standard bosons. Therefore, as long as we only consider a single mode, coherent states for anyonic bosons have exactly the same properties as those of standard bosons. Let us consider some of these properties before moving on to the distinctions between the two types of states, which arise in the multi-mode setting.
For the rest of this subsection, all symbolic expressions have the same form for both types of particles and so we use the letter $\beth$, to stand in for both $b$ and $\beta$, whenever possible. 

Eigenstates of the annihilation operator $\hat{\beth}$ can be created from the vacuum by the action of the displacement operator
\begin{equation}
\hat{D}(g)=\exp{g\hat{\beth}^{\dagger}-g^{*}\hat{\beth}},
\end{equation}
which can be written in the equivalent forms
\begin{subequations}\label{eq-60}
\begin{align}
\hat{D}(g)&=e^{-\frac{1}{2}\abs{g}^{2}}\exp{g\hat{\beth}^{\dagger}}\exp{-g^{*}\hat{\beth}},\\
\hat{D}(g)&=e^{\frac{1}{2}\abs{g}^{2}}\exp{-g^{*}\hat{\beth}}\exp{g\hat{\beth}^{\dagger}}.
\end{align}
\end{subequations}
Note that $\hat{D}(g)$ is a unitary operator, with $\hat{D}^{\dagger}(g)=\hat{D}(-g)$. Several properties of displacement operators can be derived from those identities. The most important  is that these operators ``displace'' the vacuum state. This follows from the equations
\begin{subequations}\label{eq-61}
\begin{align}
\hat{D}(-g)\hat{\beth}\hat{D}(g)&=\hat{\beth}+g\\
\hat{D}(-g)\hat{\beth}^{\dagger}\hat{D}(g)&=\hat{\beth}^{\dagger}+g^{*},
\end{align}
\end{subequations}
called the displacement identities, from which we can show that the state
\begin{equation}
\ket{g}_{\beth}=\hat{D}(g)\ket{0}_{\beth}=e^{-\frac{1}{2}\abs{g}^{2}}\sum_{n}\frac{(g\hat{\beth}^{\dagger})^{n}}{n!}\ket{0}_{\beth}
\end{equation}
is, in fact, an eigenstate of $\hat{\beth}$.

It can also be shown that single-mode coherent states of both standard and anyonic bosons satisfy a minimum uncertainty relation with respect to the quadrature operators
\begin{align}
\hat{q}&=\frac{1}{2}(\hat{\beth}^{\dagger}+\hat{\beth})\\
\hat{p}&=\frac{1}{2i}(\hat{\beth}^{\dagger}-\hat{\beth}).
\end{align}
The quadrature operators satisfy the commutation relation $[\hat{q},\hat{p}]=1$, and form a representation of the quantum harmonic oscillator. Only eigenstates of $\hat{\beth}$ can be simultaneously single-mode coherent states and minimum uncertainty states \cite{bialynickabirula_properties_1968}, which is a property that sets them apart from more general coherent states.

We finish our review by showing that displacement operators form an algebra, given by the relation
\begin{equation}
\hat{D}(g)\hat{D}(h)=e^{gh^{*}-hg^{*}}D(g+h).
\end{equation}
From this algebra, one can calculate the overlap function
\begin{equation}
\braket{h}{g}_{\beth}=e^{-\frac{1}{2}(\abs{g}^{2}+\abs{h}^{2}-2{g}h^{*})},
\end{equation}
which shows that coherent states are not orthogonal in general. Nonetheless they still satisfy the relation
\begin{equation}
\int_{\mathbf{C}}\frac{d^{2}g}{\pi}\ketbra{g}_{\beth}=\mathbb{I},
\end{equation}
with $\mathbb{I}$ as the identity operator, which makes them an over-complete basis for the single-mode state space.

\subsection{Coherent states in optical networks}\label{sec:IV-2}
We now shift to discussing the differences between bosonic anyon coherent states and their standard boson counterpart. These differences only arise in the multimode case, and we consider, for simplicity, only two modes, which we label 1 and 2.

Consider first the case of standard bosons. We denote two-mode coherent states as follows
\begin{equation}
\ket{u;v}_{b}=\hat{D}_{1}(u)\hat{D}_{2}(v)\ket{0}_{b},
\end{equation}
for any $u,v\in\mathbb{C}$ and recall, $\hat{D}$ are displacement operators.

Suppose now that the system is initialized in either of the coherent states
\begin{subequations}
\begin{align*}
\ket{g;0}_{b}&=\hat{D}_{1}(g)\ket{0}_{b},\\
\ket{0;h}_{b}&=\hat{D}_{2}(h)\ket{0}_{b},
\end{align*}
\end{subequations}
The action of a phase shifter $PS_{1}(\tau)$ on $\ket{g;0}_{b}$ is simply given by 
\begin{equation*}
PS_{1}(\tau)\ket{g;0}_{b}=\ket{ge^{i\tau};0}_{b},
\end{equation*}
and the action of $PS_{2}(\tau)$ on $\ket{0;h}_{b}$ is, similarly, given by $\ket{0;he^{i\tau}}$. 

More interesting is the action of the beam splitter $BS_{12}(\theta)$, which can be obtained using Eq.\ \  (\ref{eq-61}), resulting in
\begin{subequations}
\begin{align*}
BS_{12}(\theta)\ket{g;0}_{b}&=\ket{\cos{(\theta)}g;i\sin{(\theta)}g}_{b},\\
BS_{12}(\theta)\ket{0;h}_{b}&=\ket{i\sin{(\theta)}h;\cos{(\theta)}h}_{b}.
\end{align*}
\end{subequations}
It is easy to see that this state is two-mode coherent, in the sense that $c^{1}_{b}(n)=1$ and $c^{2}_{b}(n)=1$ for all $n$, where $c^{i}$ is the natural generalization of Eq.\ (\ref{eq:coherences}) to mode $i$. We call any state of the form
\begin{equation}
\ket{\{g_{i}\}}_{b}=\prod_{i=1}^{m}\hat{D}_{i}(g_{i})\ket{0}_{b},
\end{equation}
an exact multi-mode coherent state, or \emph{exact coherent state} for short. Notice that this state is also a simultaneous eigenstate of all $\{b_{i}\}$.

In general, the action of an arbitrary two-mode linear map $\hat{A}$ over an arbitrary two-mode coherent state is
\begin{equation*}
\hat{A}\ket{u;v}_{b}=\ket{A_{1,1}u+A_{1,2}v;A_{2,1}u+A_{2,2}v}_{b},
\end{equation*}
where $\hat{A}$ is determined by the coefficient matrix $A=[A_{i,j}]$.

It is a simple fact of linear algebra that, for any nonzero complex vector, there is a unitary matrix which rotates it into a vector with a single nonzero component. Therefore, given any $\ket{u;v}_{b}$, one can find two linear maps $\hat{A}^{1}(u,v)$ and $\hat{A}^{2}(u,v)$, such that
\begin{subequations}
\begin{align*}
\hat{A}^{1}(u,v)\ket{u;v}_{b}&=\ket{\sqrt{\abs{u}^{2}+\abs{v}^{2}};0}_{b}\\
\hat{A}^{2}(u,v)\ket{u;v}_{b}&=\ket{0;\sqrt{\abs{u}^{2}+\abs{v}^{2}}}_{b}.
\end{align*}
\end{subequations}
This observation motivates the following definition. When, for a bosonic state $\ket{\psi}_{b}$, a linear dynamic $\hat{A}(\psi)$ can be found such that $\hat{A}(\psi)\ket{\psi}_{b}$ is a single-mode coherent state, we say that $\ket{\psi}_{b}$ is a \emph{dynamically coherent} state. It is not hard to see that all dynamically coherent states are exactly coherent.

These two notions of coherent state---dynamic and exact---are not standard in the quantum optics literature, since there they coincide. However, for bosonic anyons, annihilation operators do not commute. This suggests that these definitions might not agree, as we now show.

Exact single-mode coherent states for bosonic anyons have the same form as standard bosonic ones
\begin{subequations}
\begin{align}
\ket{g;0}_{\beta}&=\hat{D}_{1}(g)\ket{0}_{\beta},\\
\ket{0;h}_{\beta}&=\hat{D}_{2}(h)\ket{0}_{\beta},
\end{align}
\end{subequations}

The action of a phase shifter on these states is similar to the standard case. The action of a beam splitter operator can be calculated from the propagation identities of Eqs.\ (\ref{eq-32a})-(\ref{eq-32b}), leading  to
\begin{subequations}
\begin{align*}
N_{g}\sum_{n}\frac{1}{n!}\prod_{k=0}^{n-1}(g\cos(\theta)\hat{\beta}^{\dagger}_{1}+ie^{-ik\varphi}g\sin(\theta)\hat{\beta}^{\dagger}_{2})\ket{0}_{\beta},\\
N_{h}\sum_{n}\frac{1}{n!}\prod_{k=0}^{n-1}(ie^{ik\varphi}h\sin(\theta)\hat{\beta}^{\dagger}_{1}+h\cos(\theta)\hat{\beta}^{\dagger}_{2})\ket{0}_{\beta},
\end{align*}
\end{subequations}
where $N_{x}=e^{-\frac{1}{2}\abs{x}^{2}}$for any $x$. We can get rid of the binomial product by using the deformed binomial identity in Eqs.\ (\ref{eq-A18}). This gives us
\begin{subequations}\label{eq43}
\begin{equation}
\left[\sum_{l,k}N_{g}c^{1}_{l,k}\frac{(\cos(\theta)g\hat{\beta}^{\dagger}_{1})^{l}}{l!}\frac{(i\sin(\theta)g\hat{\beta}^{\dagger}_{2})^{k}}{k!}\right]\ket{0}_{\beta},
\end{equation}
\begin{equation}
\left[\sum_{l,k}N_{h}c^{2}_{l,k}\frac{(i\sin{(\theta)}h\hat{\beta}^{\dagger}_{1})^{l}}{l!}\frac{(\cos{(\theta)}h\hat{\beta}^{\dagger}_{2})^{k}}{k!}\right]\ket{0}_{\beta},
\end{equation}
\end{subequations}
where we have
\begin{subequations}
\begin{align*}
c^{1}_{l,k}&=e^{-i\varphi(lk+\frac{k(k-1)}{2})},\\
c^{2}_{l,k}&=e^{i\varphi\frac{l(l-1)}{2}}.
\end{align*}
\end{subequations}
Due to the group property of beam splitters, the form of the states in Eqs.\ (\ref{eq43}) stay the same, even after successive applications of beam splitters with different angles. This observation allows us to find at least two generalizations of the notion of dynamical coherence, represented by the equations
\begin{subequations}
\begin{align*}
\ket{u;v}^{1}_{\beta}&=N\sum_{n}\frac{1}{n!}\prod_{k=0}^{n-1}(u\hat{\beta}^{\dagger}_{1}+e^{-ik\varphi}b\hat{\beta}^{\dagger}_{2})\ket{0}_{\beta},\\
\ket{u;v}^{2}_{\beta}&=N\sum_{n}\frac{1}{n!}\prod_{k=0}^{n-1}(e^{ik\varphi}u\hat{\beta}^{\dagger}_{1}+v\hat{\beta}^{\dagger}_{2})\ket{0}_{\beta},    
\end{align*}
\end{subequations}
where $N=\exp{-(\abs{u}^{2}+\abs{v}^{2})}$. We refer to these states as type 1 and type 2 dynamically coherent states, respectively.
The action of a general two-mode interferometer $\hat{A}$ is then
\begin{equation*}
\hat{A}\ket{u;v}^{i}_{\beta}=\ket{(A_{11}u+A_{12}v);(A_{21}u+A_{22}v)}^{i}_{\beta},
\end{equation*}
for $i=1,2$. 

By the definition, its easy to see that the unitary $\hat{T}=\exp{i\varphi\hat{K}}$ with $K$ given by the anyonic Kerr Hamiltonian
\begin{equation*}
\hat{K}_{12}=\frac{(\hat{n}_{1}+\hat{n}_{2})(\hat{n}_{1}+\hat{n}_{2}-1)}{2},
\end{equation*}
is such that
\begin{equation*}
\ket{u;v}^{2}_{\beta}=\hat{T}\ket{u;v}^{1}_{\beta}.
\end{equation*}
Therefore, there is no passive quadratic Hamiltonian that can convert a type $1$ dynamically coherent state into a type $2$ dynamically coherent state. 

In the case of exactly coherent two-mode states
\begin{subequations}
\begin{align*}
\ket{u;v}^{<}_{\beta}&=\hat{D}_{1}(u)\hat{D}_{2}(v)\ket{0}_{\beta},\\
\ket{u;v}^{>}_{\beta}&=\hat{D}_{2}(v)\hat{D}_{1}(u)\ket{0}_{\beta},
\end{align*}
\end{subequations}
notice that they can be written as
\begin{subequations}
\begin{align*}
\left[\sum_{l,k}N\frac{(u\hat{\beta}^{\dagger}_{1})^{l}}{l!}\frac{(v\hat{\beta}^{\dagger}_{2})^{k}}{k!}\right]&\ket{0}_{\beta},\\
\left[\sum_{l,k}Ne^{-i\varphi kl}\frac{(u\hat{\beta}^{\dagger}_{1})^{l}}{l!}\frac{(v\hat{\beta}^{\dagger}_{2})^{k}}{k!}\right]&\ket{0}_{\beta}.
\end{align*}
\end{subequations}
By inspection, we see that no exactly coherent state can be mapped into a dynamically coherent states using passive quadratic anyonic Hamiltonians.

We believe that this incursion into the different kinds of two-mode coherent states of bosonic anyons is enough to illustrate the drastic effects the anyonic exchange phase has on coherent state dynamics. As a last example, let us study the effect of a mirror, i.e.\  the network given by $PS_{1}(\pi/2)BS_{12}(\pi/2)PS_{2}(\pi/2)$, on single-mode coherent states:
\begin{subequations}
\begin{align*}
\ket{u;0}_{\beta}&=\hat{D}_{1}(u)\ket{0}_{\beta},\\
\ket{0;v}_{\beta}&=\hat{D}_{2}(v)\ket{0}_{\beta}.
\end{align*}
\end{subequations}

From our previous discussion it follows that the output states are given by
\begin{subequations}
\begin{align*}
\ket{0;u}^{1}_{\beta}&=N_{u}\sum_{k}e^{-i\varphi\frac{k(k-1)}{2}}\frac{(u\hat{\beta}^{\dagger}_{2})^{k}}{k!}\ket{0;0},\\
\ket{v;0}^{2}_{\beta}&=N_{v}\sum_{k}e^{i\varphi\frac{k(k-1)}{2}}\frac{(v\hat{\beta}^{\dagger}_{1})^{k}}{k!}\ket{0;0},
\end{align*}
\end{subequations}
where the first state is a reflection of $\ket{u;0}_{\beta}$ and the second a reflection of $\ket{0;v}_{\beta}$. 

For all values of $\varphi$, we can use the prescription in \cite{glauber_coherent_1963} to define the coherent-basis wave function 
\begin{equation*}
\Psi_{i}^{\varphi}(x,z)=\sum_{k}\psi^{i}_{k}(x)\frac{z^{k}}{\sqrt{k!}},
\end{equation*}
which, in our case, has
\begin{subequations}
\begin{equation*}
\psi^{1}_{k}(x)=N_{x}c^{*}_{k}\frac{\hat{x}^{k}}{\sqrt{k!}},
\end{equation*}
and
\begin{equation*}
\psi^{2}_{k}(x)=N_{x}c_{k}\frac{\hat{x}^{k}}{\sqrt{k!}},
\end{equation*}
\end{subequations}
and write the output states as a linear combination of the appropriate single-mode coherent states, in the form
\begin{subequations}
\begin{align*}
\ket{0;u}^{1}_{\beta}&=\frac{1}{\pi}\int_{\mathbf{C}}d^{2}z\ket{0;z}_{\beta}\Psi^{2}(b,z^{*})e^{-\frac{1}{2}\abs{z}^{2}},\\
\ket{v;0}^{2}_{\beta}&=\frac{1}{\pi}\int_{\mathbf{C}}d^{2}z\ket{z;0}_{\beta}\Psi^{1}(a,z^{*})e^{-\frac{1}{2}\abs{z}^{2}}.
\end{align*}
\end{subequations}

For further illustration, let us take the simplest case $\varphi=\pi$. We see that the mirror acts as
\begin{subequations}
\begin{align*}
\ket{0;u}^{1}_{\beta}&=\frac{1}{N\sqrt{2}}\left[(-1)^{\frac{1}{4}}\ket{0;-i u}_{\beta}-(-1)^{\frac{3}{4}}\ket{0;i u}_{\beta}\right],
\\
\ket{v;0}^{2}_{\beta}&=\frac{1}{N\sqrt{2}}\left[(-1)^{\frac{1}{4}}\ket{-i v;0}_{\beta}-(-1)^{\frac{3}{4}}\ket{i v;0}_{\beta}\right],
\end{align*}
\end{subequations}
up to a normalization factor $N$. Such states are called cat states, and they have multiple applications in quantum information theory, such as encoding logical qubits or as a resource for teleportation protocols \cite{jeong_efficient_2002,ralph_quantum_2003,gilchrist_schrodinger_2004,lund_fault-tolerant_2008,menicucci_fault-tolerant_2014,mirrahimi_dynamically_2014,ketterer_quantum_2016,lau_universal_2016}.

\section{Conclusion}\label{V}
We generalized the formalism of linear-optical networks, commonly applied to standard bosons and fermions, to fermionic and bosonic anyons on a 1D lattice. We showed that anyonic optical networks cannot be uniquely characterized by the matrix of single-mode transition amplitudes, in contrast to standard bosons and fermions. This is due, in part, to the existence of special anyonic optical networks, which take advantage of Aharonov-Bohm phases that arise in these anyonic systems. We also showed that the dynamics induced by anyonic optical elements preserve the characteristic bosonic and fermionic bunching behavior in the form of the Hong-Ou-Mandel effect and Pauli exclusion principle, respectively. 

From the anyonic linear-optical dynamics we were able to propose a deterministic entangling two-qubit gate applicable to both the bosonic and fermionic cases. This proposal requires one auxiliary mode populated with a single particle as a (reusable) resource, and generalizes previous results \cite{tosta_quantum_2019}. 

We also showed how to define coherent states of bosonic anyons, describing how these states can be classified into families that are not equivalent up to linear optical transformations. Finally, as an application, we discussed how an anyonic mirror acting on a single-mode coherent state can create cat states, which are often considered a valuable resource for quantum information processing in continuous variables. 

We believe the anyonic characteristics we have uncovered reveal interesting features on the uniqueness of standard fermionic and bosonic behavior. This framework might also be appropriate to study the role of fractional statistics over general types of quantum information tasks, as well as measures for resource states over such tasks. We also hope that this work can serve as a guide to building a framework for optical network models of more complex kinds of anyons, such as Fock parafermions \cite{cobanera_fock_2014} and non-abelian anyon ladders \cite{feiguin_interacting_2007}.

\section{Acknowledgements}
This work was supported by project Instituto Nacional de Ciência e Tecnologia de Informação Quântica (INCT-IQ/CNPq). EFG acknowledges the Portuguese funding institution FCT – Fundação para Ciência e Tecnologia via project CEECINST/00062/2018.
\appendix

\section{Mathematical proofs}
Here we give the proofs for various statements made throughout this paper.

\subsection{SU(2) algebra of quadratic anyonic operators} \label{app:su2}
\begin{theorem}
Define the operators
\begin{subequations}
\begin{align}
&J^{1}_{ij}=\frac{1}{2}(\hat{\chi}^{\dagger}_{i}\hat{\chi}_{j}+\hat{\chi}^{\dagger}_{j}\hat{\chi}_{i}),\\
&J^{2}_{ij}=\frac{-i}{2}(\hat{\chi}^{\dagger}_{i}\hat{\chi}_{j}-\hat{\chi}^{\dagger}_{j}\hat{\chi}_{i}),\\
&J^{3}_{ij}=\frac{1}{2}(\hat{\chi}^{\dagger}_{i}\hat{\chi}_{i}-\hat{\chi}^{\dagger}_{j}\hat{\chi}_{j}),
\end{align}
\end{subequations}
for fixed $i$ and $j$. Then
\begin{equation}
[J^{k}_{ij};J^{l}_{ij}]=i\epsilon_{klm}J^{m}_{ij},
\end{equation}
for all $k,l,m=1,...,3$
\end{theorem}

\begin{proof}
We just have to compute the three commutators $[J^{1}_{ij};J^{2}_{ij}]$, $[J^{1}_{ij};J^{3}_{ij}]$ and $[J^{2}_{ij};J^{3}_{ij}]$. The second and third ones are trivial, given that $\hat{\chi}_{i}^{\dagger}\hat{\chi}_{i}$ are number operators for their respective modes. Therefore, we only have to compute $[J^{1}_{ij};J^{2}_{ij}]$.

Expanding the terms in $[J^{1}_{ij};J^{2}_{ij}]$ we have that $J^{1}_{ij}J^{2}_{ij}$ is given by
\small
\begin{equation*}
\frac{-i}{4}\left[(\hat{\chi}^{\dagger}_{i}\hat{\chi}_{j})^{2}-(\hat{\chi}^{\dagger}_{j}\hat{\chi}_{i})^{2}\right]-\left[\hat{\chi}^{\dagger}_{i}(\hat{\chi}_{j}\hat{\chi}^{\dagger}_{j})\hat{\chi}_{i}-\hat{\chi}^{\dagger}_{j}(\hat{\chi}_{i}\hat{\chi}^{\dagger}_{i})\hat{\chi}_{j}\right],
\end{equation*}
\normalsize
while $J^{2}_{ij}J^{1}_{ij}$ is given by
\small
\begin{equation*}
\frac{-i}{4}\left[(\hat{\chi}^{\dagger}_{i}\hat{\chi}_{j})^{2}-(\hat{\chi}^{\dagger}_{j}\hat{\chi}_{i})^{2}\right]+\left[\hat{\chi}^{\dagger}_{i}(\hat{\chi}_{j}\hat{\chi}^{\dagger}_{j})\hat{\chi}_{i}-\hat{\chi}^{\dagger}_{j}(\hat{\chi}_{i}\hat{\chi}^{\dagger}_{i})\hat{\chi}_{j}\right].
\end{equation*}
\normalsize
Using the commutation relations for either bosonic or fermionic anyons we obtain
\begin{align*}
[J^{1}_{ij};J^{2}_{ij}]&=\frac{i}{2}\left[\hat{\chi}^{\dagger}_{i}\hat{\chi}_{i}(\hat{\chi}_{j}\hat{\chi}^{\dagger}_{j})-\hat{\chi}^{\dagger}_{j}\hat{\chi}_{j}(\hat{\chi}_{i}\hat{\chi}^{\dagger}_{i})\right] \notag\\
&=\frac{i}{2}(\hat{\chi}^{\dagger}_{i}\hat{\chi}_{i}-\hat{\chi}^{\dagger}_{j}\hat{\chi}_{j})=iJ^{3
}_{ij},
\end{align*}
as desired.
\end{proof}

\subsection{Propagation identities}\label{app:propagation}
By using the $SU(2)$ structure of bilinear operators we can effectively solve the dynamics of creation and annihilation operators for bosonic anyons under the action of a beam splitter.
\begin{theorem}
Let $G$ be given by
\begin{equation*}
\hat{G}^{n\varphi}_{ij}(\theta)=e^{in\varphi{J}^{3}_{ij}}e^{i\theta(2J^{1}_{ij})}e^{-in\varphi{J}^{3}_{ij}}.
\end{equation*}
then it follows that
\begin{subequations}
\begin{align*}
\hat{G}^{n\varphi}_{ij}(\theta)\hat{\beta}^{\dagger}_{i}=&(\cos{\theta}\hat{\beta}^{\dagger}_{i}+ie^{-in\varphi}\sin{\theta}\hat{\beta}^{\dagger}_{j})\hat{G}^{(n+1)\varphi}_{ij}(\theta),\\
\hat{G}^{n\varphi}_{ij}(\theta)\hat{\beta}^{\dagger}_{j}=&(\cos{\theta}\hat{\beta}^{\dagger}_{j}+ie^{in\varphi}\sin{\theta}\hat{\beta}^{\dagger}_{i})\hat{G}^{(n+1)\varphi}_{ij}(\theta)
\end{align*}
and if $i<k<j$
\begin{equation*}
\hat{G}^{n\varphi}_{ij}(\theta)\hat{\beta}^{\dagger}_{k}=\hat{\beta}^{\dagger}_{k}\hat{G}^{(n+2)\varphi}_{ij}(\theta).
\end{equation*}
\end{subequations}
\end{theorem}
\begin{proof}
The $i<k<j$ case is easy to see from the calculation for fermionic anyons. For the cases when $k=i$ or $k=j$, let us compute the commutator between the beam splitter Hamiltonian ${H}^{BS}_{ij}=\hat{\beta}^{\dagger}_{i}\hat{\beta}_{j}+\hat{\beta}^{\dagger}_{j}\hat{\beta}_{i}=2{J}^{1}_{ij}$ and the creation operators $\hat{\beta}^{\dagger}_{i}$ and $\hat{\beta}^{\dagger}_{j}$
\begin{subequations}
\begin{align*}
[(2{J}^{1}_{ij}),\hat{\beta}^{\dagger}_{i}]=&\hat{\beta}^{\dagger}_{i}\left[(\cos{\varphi}-1)(2{J}^{1}_{ij})-\sin{\varphi}(2{J}^{2}_{ij})\right]+\hat{\beta}^{\dagger}_{j},\\
[(2{J}^{1}_{ij}),\hat{\beta}^{\dagger}_{j}]=&\hat{\beta}^{\dagger}_{j}\left[(\cos{\varphi}-1)(2{J}^{1}_{ij})-\sin{\varphi}(2{J}^{2}_{ij})\right]+\hat{\beta}^{\dagger}_{i},
\end{align*}
\end{subequations}
this implies that
\begin{subequations}
\begin{align*}
(2{J}^{1}_{ij})\hat{\beta}^{\dagger}_{i}=&\hat{\beta}^{\dagger}_{i}\left[2(\cos{\varphi}{J}^{1}_{ij}-\sin{\varphi}{J}^{2}_{ij})\right]+\hat{\beta}^{\dagger}_{j},\\
(2{J}^{1}_{ij})\hat{\beta}^{\dagger}_{j}=&\hat{\beta}^{\dagger}_{j}\left[2(\cos{\varphi}{J}^{1}_{ij}-\sin{\varphi}{J}^{2}_{ij})\right]+\hat{\beta}^{\dagger}_{i}.
\end{align*}
\end{subequations}
Combining the two equations above we get
\begin{subequations}
\begin{align*}
(2{J}^{1}_{ij})(\hat{\beta}^{\dagger}_{i}+\hat{\beta}^{\dagger}_{j})=&(\hat{\beta}^{\dagger}_{i}+\hat{\beta}^{\dagger}_{j})\left[2(\cos{\varphi}{J}^{1}_{ij}-\sin{\varphi}{J}^{2}_{ij})+1\right],\\
(2{J}^{1}_{ij})(\hat{\beta}^{\dagger}_{i}-\hat{\beta}^{\dagger}_{j})=&(\hat{\beta}^{\dagger}_{i}-\hat{\beta}^{\dagger}_{j})\left[2(\cos{\varphi}{J}^{1}_{ij}-\sin{\varphi}{J}^{2}_{ij})-1\right].
\end{align*}
\end{subequations}
Multiplying equations above by powers of $(2{J}^{1}_{ij})$ we find the  expression for commuting a beam splitter with a combination of creation operators
\begin{subequations}
\begin{align*}
e^{i\theta(2{J}^{1}_{ij})}(\hat{\beta}^{\dagger}_{i}+\hat{\beta}^{\dagger}_{j})&=(\hat{\beta}^{\dagger}_{i}+\hat{\beta}^{\dagger}_{j})e^{i\theta\left[2(\cos{\varphi}{J}^{1}_{ij}-\sin{\varphi}{J}^{2}_{ij})+1\right]}, \\
e^{i\theta(2{J}^{1}_{ij})}(\hat{\beta}^{\dagger}_{i}-\hat{\beta}^{\dagger}_{j})&=(\hat{\beta}^{\dagger}_{i}-\hat{\beta}^{\dagger}_{j})e^{i\theta\left[2(\cos{\varphi}{J}^{1}_{ij}-\sin{\varphi}{J}^{2}_{ij})-1\right]}.
\end{align*}
\end{subequations}
Finally, this leads to
\begin{subequations}
\begin{align*}
e^{i\theta(2{J}^{1}_{ij})}\hat{\beta}^{\dagger}_{i}&=(\cos{\theta}\hat{\beta}^{\dagger}_{i}+i\sin{\theta}\hat{\beta}^{\dagger}_{j})e^{i\theta\left[2(\cos{\varphi}{J}^{1}_{ij}-\sin{\varphi}{J}^{2}_{ij})\right]},\\
e^{i\theta(2{J}^{1}_{ij})}\hat{\beta}^{\dagger}_{j}&=(i\sin{\theta}\hat{\beta}^{\dagger}_{i}+\cos{\theta}\hat{\beta}^{\dagger}_{j})e^{i\theta\left[2(\cos{\varphi}{J}^{1}_{ij}-\sin{\varphi}{J}^{2}_{ij})\right]}.
\end{align*}
\end{subequations}
To cast this result in a more illuminating form, recall that
\begin{equation*}
\cos{\varphi}{J}^{1}_{ij}-\sin{\varphi}{J}^{2}_{ij}=e^{i\varphi{J}^{3}_{ij}}{J}^{1}_{ij}e^{-i\varphi{J}^{3}_{ij}},
\end{equation*}
and, therefore,
\begin{equation*}
e^{i\theta\left[2(\cos{\varphi}{J}^{1}_{ij}-\sin{\varphi}{J}^{2}_{ij}))\right]}=e^{i\varphi{J}^{3}_{ij}}e^{i\theta(2{J}^{1}_{ij})}e^{-i\varphi{J}^{3}_{ij}},
\end{equation*}
which gives us
\begin{subequations}
\begin{align*}
\hat{G}^{0}_{ij}(\theta)\hat{\beta}^{\dagger}_{i}=&(\cos{\theta}\hat{\beta}^{\dagger}_{i}+i\sin{\theta}\hat{\beta}^{\dagger}_{j})\hat{G}^{\varphi}_{ij}(\theta),\\
\hat{G}^{0}_{ij}(\theta)\hat{\beta}^{\dagger}_{j}=&(\cos{\theta}\hat{\beta}^{\dagger}_{j}+i\sin{\theta}\hat{\beta}^{\dagger}_{i})\hat{G}^{\varphi}_{ij}(\theta),
\end{align*}
\end{subequations}
by the definition of $G$. Then it is not hard to see, by induction on $n$, that
\begin{align*}
\hat{G}^{n\varphi}_{ij}(\theta)\hat{\beta}^{\dagger}_{i}=&(\cos{\theta}\hat{\beta}^{\dagger}_{i}+ie^{-in\varphi}\sin{\theta}\hat{\beta}^{\dagger}_{j})\hat{G}^{(n+1)\varphi}_{ij}(\theta),\\
\hat{G}^{n\varphi}_{ij}(\theta)\hat{\beta}^{\dagger}_{j}=&(\cos{\theta}\hat{\beta}^{\dagger}_{j}+ie^{in\varphi}\sin{\theta}\hat{\beta}^{\dagger}_{i})\hat{G}^{(n+1)\varphi}_{ij}(\theta),
\end{align*}
thus proving the theorem.
\end{proof}

\subsection{Generalized binomial identities}
Applying the propagation identities to powers of creation operators leads to the necessity of calculating generalized binomial identities for anyons.

\begin{theorem}
Let $a,b$ be arbitrary complex numbers and $i<j$, then
\begin{subequations}\label{eq-A18}
\begin{equation}\label{eq-A18a}
\prod_{k=0}^{n-1}(a\hat{\beta}^{\dagger}_{i}+e^{-ik\varphi}b\hat{\beta}^{\dagger}_{j})=e^{-i\varphi\frac{n(n-1)}{2}}\prod_{k=0}^{n-1}(e^{ik\varphi}a\hat{\beta}^{\dagger}_{i}+b\hat{\beta}^{\dagger}_{j}),
\end{equation}
where
\begin{equation}\label{eq-A18b}
\prod_{k=0}^{n-1}(e^{ik\varphi}a\hat{\beta}^{\dagger}_{i}+b\hat{\beta}^{\dagger}_{j})=\sum_{l=0}^{n}\binom{n}{l}e^{i\varphi\frac{l(l-1)}{2}}(a\hat{\beta}^{\dagger}_{i})^{l}(b\hat{\beta}^{\dagger}_{j})^{n-l}.
\end{equation}
\end{subequations}
\end{theorem}

\begin{proof}
The first equality is easy to prove, since
\begin{equation*}
a\hat{\beta}^{\dagger}_{i}+e^{-ik\varphi}b\hat{\beta}^{\dagger}_{j}=e^{-ik\varphi}(e^{ik\varphi}a\hat{\beta}^{\dagger}_{i}+b\hat{\beta}^{\dagger}_{j}),
\end{equation*}
which implies 
\begin{equation*}
\prod_{k=0}^{n-1}(a\hat{\beta}^{\dagger}_{i}+e^{-ik\varphi}b\hat{\beta}^{\dagger}_{j})=\left[\prod_{k=0}^{n-1}e^{-i\varphi k}\right]\prod_{k=0}^{n-1}(e^{ik\varphi}a\hat{\beta}^{\dagger}_{i}+b\hat{\beta}^{\dagger}_{j}).
\end{equation*}
Then, just notice that 
\begin{equation*}
\prod_{k=0}^{n-1}e^{-i\varphi k}=e^{i\varphi\sum_{n-1}^{k=0}k}=e^{i\varphi\frac{n(n-1)}{2}}.
\end{equation*}

We prove the second identity using induction over $n$. The identity is trivial for $n=1$. Suppose it is valid for the $n$th case. Then the $n+1$th case can be written as 
\begin{equation*}
\sum_{l=0}^{n}\binom{n}{l}e^{i\varphi\frac{l(l-1)}{2}}(a\hat{\beta}^{\dagger}_{i})^{l}(b\hat{\beta}^{\dagger}_{j})^{n-l}(e^{in\varphi}a\hat{\beta}^{\dagger}_{i}+b\hat{\beta}^{\dagger}_{j}).
\end{equation*}
Expanding the factors and rearranging the operators such that they are normally ordered we obtain
\begin{multline*}
\sum_{l=0}^{n}\binom{n}{l}e^{i\varphi\frac{l(l+1)}{2}}(a\hat{\beta}^{\dagger}_{i})^{l+1}(b\hat{\beta}^{\dagger}_{j})^{n-l}+\\
+\sum_{l=0}^{n}\binom{n}{l}e^{i\varphi\frac{l(l-1)}{2}}(a\hat{\beta}^{\dagger}_{i})^{l}(b\hat{\beta}^{\dagger}_{j})^{n-l+1}.
\end{multline*}
We proceed by separating terms $l=n$ from the first sum and $l=0$ from the second one, writing them explicitly. After that, we relabel $k=l+1$ in the first sum and $k=l$ in the second, obtaining
\begin{multline*}
e^{i\varphi\frac{n(n+1)}{2}}(a\hat{\beta}^{\dagger}_{i})^{n+1}+\\
+\sum_{k=1}^{n}\left[\binom{n}{k-1}+\binom{n}{k}\right]e^{i\varphi\frac{k(k-1)}{2}}(a\hat{\beta}^{\dagger}_{i})^{k}(b\hat{\beta}^{\dagger}_{j})^{(n+1)-k}+\\
+(b\hat{\beta}^{\dagger}_{j})^{n+1}.
\end{multline*}
Finally we use that
\begin{equation*}
\binom{n}{k-1}+\binom{n}{k}=\binom{n+1}{k},
\end{equation*}
from which it follows the identity for the $n+1$th case, completing the induction.
\end{proof}
\end{document}